\documentclass[journal]{IEEEtran}

\usepackage{cite}
\usepackage{amsmath,amssymb,amsfonts}

\usepackage[utf8]{inputenc}
\usepackage{xspace}
\usepackage{balance}
\usepackage{booktabs}
\usepackage{tabularx}
\usepackage{graphicx}
\usepackage{xcolor}
\usepackage{hyperref}

\usepackage{thmtools}
\usepackage{thm-restate}
\usepackage[nameinlink,capitalise,noabbrev]{cleveref}
\newtheorem{lemma}{Lemma}

\newtheorem{theorem}{Theorem}
\newtheorem{corollary}{Corollary}

\newtheorem{claim}{Claim}
\newtheorem{observation}{Observation}

\crefname{claim}{Claim}{Claims}
\crefname{observation}{Observation}{Observations}
\crefname{fact}{Fact}{Facts}
\crefname{property}{Property}{Properties}
\crefname{invariant}{Invariant}{Invariants}

\newcommand{\paperfigure}[2][]{%
  \IfFileExists{#2.pdf}{\includegraphics[#1]{#2}}{%
  \IfFileExists{#2.png}{\includegraphics[#1]{#2}}{%
  \IfFileExists{#2.jpg}{\includegraphics[#1]{#2}}{%
  \IfFileExists{#2.jpeg}{\includegraphics[#1]{#2}}{%
  \IfFileExists{#2.eps}{\includegraphics[#1]{#2}}{%
    \fbox{\parbox{0.75\columnwidth}{\centering Missing figure: \texttt{\detokenize{#2}}}}%
  }}}}}%
}

\newcommand{\ALG}{\textsc{TC}\xspace}
\newcommand{\OPT}{\textsc{Opt}\xspace}
\newcommand{\last}{\textrm{last}}
\newcommand{\size}{\textrm{size}}
\newcommand{\req}{\textrm{req}}
\newcommand{\cnt}{\textrm{cnt}}
\newcommand{\val}{\textrm{val}}
\newcommand{\degree}{\textrm{deg}}
\newcommand{\beP}{\textrm{begin}(P)}
\newcommand{\enP}{\textrm{end}(P)}
\newcommand{\F}{\mathcal{F}}
\newcommand{\kALG}{k_\textnormal{ONL}}
\newcommand{\kOPT}{k_\textnormal{OPT}}
\newcommand{\pin}{\textnormal{\textsc{in}}\xspace}
\newcommand{\pout}{\textnormal{\textsc{out}}\xspace}
\newcommand{\VOPT}{V_\textnormal{OPT}}
\newcommand{\VOPTC}{V_\textnormal{OPT}^\textrm{c}}

\begin{document}

\title{Online Tree Caching}

\author{Marcin~Bienkowski, Jan~Marcinkowski, Maciej~Pacut, Stefan~Schmid, Aleksandra~Spyra, and~Bolong~Tang%
\thanks{M.~Bienkowski is supported by Polish National Science Centre grant 2022/45/B/ST6/00559.}%
\thanks{S.~Schmid is supported by European Research Council (ERC), project FortifyNet (grant 101287293), 2026-2027.}%
\thanks{M.~Bienkowski and A.~Spyra are with the Institute of Computer Science, University of Wroc{\l}aw, Wroc{\l}aw, Poland.}%
\thanks{J.~Marcinkowski is with Google Poland, Wroc{\l}aw, Poland. His work was performed while he was at University of Wroc{\l}aw.}%
\thanks{M.~Pacut and S.~Schmid are with TU Berlin, Berlin, Germany.}%
\thanks{B.~Tang is with KTH Royal Institute of Technology, Stockholm, Sweden, and with TU Berlin, Berlin, Germany.}}


\maketitle

\begin{abstract}
We initiate the study of a natural and practically relevant new variant of online caching where the to-be-cached items can have dependencies.  We assume that the universe is a~tree~$T$ and items are tree nodes; we require that if a node $v$ is cached then the whole subtree $T(v)$ rooted at $v$ is cached as well. This theoretical problem finds an immediate application in the context of forwarding table optimization in IP routing and software-defined networks.

We present an elegant online deterministic algorithm \ALG for this problem, and rigorously prove that its competitive ratio is $O(\textsc{height}(T) \cdot \kALG/(\kALG-\kOPT+1))$, where $\kALG$ and $\kOPT$ denote the cache sizes of an online and the optimal offline algorithm, respectively. The result is optimal up to a factor of $O(\textsc{height}(T))$. We complement the theoretical analysis with an~empirical study on synthetic workloads and CAIDA backbone traffic, illustrating the trade-off between hit rate and cache-update cost.
\end{abstract}


\section{Introduction}

In the classic online paging problem, items of some universe are requested by a~processing entity (e.g., blocks of RAM are requested by the processor). To speed up the access, computers use faster memory, called a~\emph{cache}, capable of accommodating~$k$ such items. Upon a~request to a non-cached item, the algorithm has to fetch it into the cache, paying a~fixed cost, while a request to a cached item is free. If the cache is full, the algorithm has to free some space by evicting an~arbitrary subset of items from the cache.

The paging problem is inherently online: the algorithm has to decide what to evict from the cache without knowing the future requests; its cost is compared to the cost of an optimal \emph{offline} solution and the worst-case ratio of these two amounts is called the \emph{competitive ratio}. The first analysis of this basic problem in an online model was given over four decades ago by Sleator and Tarjan~\cite{competitive-analysis}. The problem was later considered in a~variety of flavors. In particular, some papers considered a \emph{bypassing model}~\cite{caching-rejection-penalties,paging-irani}, where item fetching is optional: the requested item can be served without being in the cache, for another fixed cost (usually at most the cost of fetching the item).

In this paper, we introduce a natural extension of this fundamental problem, where items have inter-dependencies. More precisely, we assume that the universe is an arbitrary (not necessarily binary) rooted tree $T$ and the requested items are its nodes. For any tree node $v$, $T(v) \subseteq T$ is a subtree rooted at $v$ containing $v$ and all its descendants. We require the following property: if a~node $v$ is in the cache, then all nodes of $T(v)$ are also cached. In other words, we require that \emph{the cache is a~subforest of $T$}, i.e., a union of disjoint subtrees of~$T$.  We call this problem \emph{online tree caching}.

Furthermore, we assume a bypassing model and distinguish between two types of requests: a request can be either \emph{positive} or \emph{negative}. The positive requests correspond to ``normal'' requests known from caching problems: we pay~$1$ if the node is not cached; for a negative request, we pay $1$ if the requested node is cached. After serving the request, we may reorganize our cache arbitrarily, but the resulting cache has to still be a subforest of $T$. We pay $\alpha$ for fetching or evicting any single node, where $\alpha \geq 1$ is an integer and a~parameter of the problem. Our goal is to minimize the overall cost of maintaining the cache and serving the requests.

One interesting application for our model arises in the context of modern IP routers which need to store a rapidly increasing number of forwarding rules~\cite{bgp-routeviews,steve-myth}. In \cref{sec:motivation}, we outline this application, discussing how tree caching algorithms can be applied in existing systems to effectively reduce the memory requirements on IP routers.


\subsection{Our Contributions and Paper Organization}

We initiate the study of a natural new variant of caching with bypassing, which accounts for tree dependencies among items. The problem finds immediate applications, e.g., in IP routing and software-defined networking (see \cref{sec:motivation}).

In particular, we consider the online tree caching problem within the resource augmentation paradigm: we assume that cache sizes of the online algorithm ($\kALG$)  and the optimal offline algorithm ($\kOPT$) may differ. We assume $\kALG \geq \kOPT$ and let $R = \kALG/(\kALG-\kOPT+1)$.

In \cref{sec:algo}, we present an elegant deterministic online algorithm~\ALG for this problem. While our algorithm is simple, its analysis presented in \cref{sec:analysis} requires several non-trivial insights into the problem. In particular, we rigorously prove that \ALG is $O(h(T) \cdot R)$-competitive, where $h(T)$ is the height of tree~$T$. That is, we show that there exists $\beta$ independent of the input, such that $\ALG(I) \leq O(h(T) \cdot R) \cdot \OPT(I) + \beta$ for any input $I$. Note that this result is optimal up to the factor~$O(h(T))$: in \cref{sec:lower-bound-on-the-problem}, we show that the lower bound~$R$ for the paging problem~\cite{competitive-analysis} implies an $\Omega(R)$ lower bound for our problem for any $\alpha \geq 1$. In \cref{sec:implementing_counters}, we show that \ALG can be implemented efficiently.

Finally, \cref{sec:experimental_implementation} describes a simulator of the studied setting and \cref{sec:experimental_evaluation} reports an experimental comparison of \ALG's positive-request path with other caching strategies on synthetic and real-world packet traces.


\subsection{Related Work on Caching}

Our formal model is a novel variant of competitive paging, a~classic online problem. In the framework of competitive analysis, the paging problem was first analyzed by Sleator and Tarjan~\cite{competitive-analysis}, who showed that algorithms \textsc{Least-Recently-Used}, \textsc{First-In-First-Out} and \textsc{Flush-When-Full} are $\kALG / (\kALG - \kOPT + 1)$-competitive and no deterministic algorithm can beat this ratio. In the non-augmented case when $\kALG = \kOPT = k$, the competitive ratio is simply $k$.

The simple paging problem was later generalized to allow different fetching costs (weighted paging)~\cite{double-coverage,young-paging-greedy-dual} and additionally different item sizes (file caching)~\cite{young-paging-landlord}, with the same competitive ratio. Asymptotically the same results can be achieved when bypassing is allowed (see \cite{caching-rejection-penalties,paging-irani} and references therein). With randomization, the competitive ratio can be reduced to $O(\log k)$ even for file caching~\cite{generalized-caching-optimal}. The lower bound for randomized algorithms is $H_k = \Theta(\log k)$~\cite{paging-mark} and is matched by known paging algorithms~\cite{paging-optimal-easy,paging-optimal-difficult}.

The variant of caching in which an~item may be fetched only together with the items it depends on was first considered within the framework of competitive analysis in the conference version of this paper~\cite{tree-caching-spaa}. Two subsequent papers study related dependency models.

Dallot et al.~\cite{dependency-aware-caching} allow the dependencies to form an~arbitrary directed acyclic graph, and they do not consider resource augmentation. Closest to our work is their treatment of the bypassing model, for which they give a randomized $O(\sqrt{k \log k})$-competitive algorithm.\footnote{Their bound is $O(\sqrt{k \cdot H_{\min\{k,\ell\}}})$, where $\ell$ is the size of the maximum independent set in the transitive closure of the dependency graph; for a tree, $\ell$ becomes the number of leaves of~$T$. We state the bound for $\ell \geq k$, which is the typical case for forwarding tables.} Their bypassing model is closest to our positive-request regime at $\alpha=1$. Their guarantee is randomized, so it does not contradict our deterministic lower bound; for deterministic algorithms with bypassing, no bound better than ours is known.

Doron-Arad and Naor~\cite{non-linear-paging} study non-linear paging, where an~arbitrary monotone function $f$ decides which cache states are feasible; their model features neither bypassing nor negative requests. Choosing $f(S) = |\bigcup_{v \in S} T(v)|$ expresses our subforest requirement, and their result then yields a $k$-competitive algorithm for the resulting instance of non-linear paging. Their cost measure applied to our model, however, would make the cost of each fetch or eviction constant, independent of the number of tree nodes that are fetched or evicted.

We finally mention several models that are related in name or in spirit, but differ in substance. In restricted caching~\cite{restricted-caching}, and in its variants matroid caching~\cite{matroid-caching} and companion caching~\cite{companion-caching}, the restriction concerns cache locations: an item may be placed only in a prescribed subset of them. In block-aware caching~\cite{block-aware-caching}, items are partitioned into blocks, and moving several items of one block at the same time costs as much as moving a single one. There, however, every subset of items still forms a legal cache. Multi-level paging~\cite{multi-level-paging}, despite its name, involves no hierarchy among the items, instead generalizing read and write requests to multiple request types.


\section{Application: Minimizing Forwarding Tables in Routers}
\label{sec:motivation}

Dependencies among to-be-cached items arise in numerous settings and are a natural refinement of many caching problems. To give a concrete example, one important application for our tree-based dependency model arises in the context of IP routers. In particular, the online tree caching problem we introduce in this paper is motivated by router memory constraints in IP-based networks. The material presented in this section serves as motivation, and is not necessary for understanding the theoretical part of the paper.

Nowadays, routers have to store an~enormous number of forwarding rules. As of August 2026, the Route-Views Oregon collector reports approximately 1.12 million active IPv4 entries, of which roughly 53\% are more-specific prefixes~\cite{bgp-routeviews}. This entails large costs for Internet Service Providers: fast router memory (usually Ternary Content Addressable Memory (TCAM)) is expensive and power-hungry~\cite{tcam-expensive}.  Many routers currently operate at (or even beyond) the edge of their memory capacities. A~solution, which could delay the need for expensive or impossible memory upgrades in routers, is to store only a subset of rules in the actual router and store all rules on a~secondary device (for example a commodity server with a large but slow memory)~\cite{cacheflow,route-caching-flat,prefix-caching,fib-caching-non-overlapping,fibium-zipf}.

This solution is particularly attractive with the advent of Soft\-ware-Defined Networking (SDN) technology, which makes it possible to manage the expensive memory using a software controller~\cite{cacheflow,fibium-zipf}. In particular, our theoretical model can describe real-world architectures like~\cite{cacheflow,fibium-zipf}, that is, our model formalizes the underlying operational problems of such architectures. Our algorithm, when applied in the context of such architectures, can hence be used to prolong the lifetime of IP routers.

\paragraph{Setup, positive requests, fetches and evictions}

The setup (see~\cite{fibium-zipf} for a more technical discussion) depicted in \cref{fig:motivation} consists of two entities: the actual router (e.g., an OpenFlow switch) which caches only a~subset of all forwarding rules, and the (SDN) controller, which keeps all rules in its less expensive and slower memory. At runtime, packets arrive at the router, and if an appropriate forwarding rule is found within the rules cached by the router, then the packet is forwarded accordingly, and the associated cost is zero. Otherwise, the packet has to be forwarded to the controller (where an~appropriate forwarding rule exists); this indirection costs~$1$. Hence, the rules correspond to cacheable items and accesses to rules are modeled by positive requests to the corresponding items. At some chosen points in time, the caching algorithm run at the controller may decide to add rules to the cache or to remove them from it. Any such change entails a~fixed cost $\alpha$.\footnote{This cost corresponds to the transmission of a message from the controller to the router as well as the update of internal data structures of the router. Such an update of proprietary and vendor-dependent structures can be quite costly~\cite{tcam-expensive-updates}, but empirical studies show it to be independent of the rule being updated~\cite{fib-updates}.}

\begin{figure}[t]
  \centering
  \paperfigure[width=0.8\columnwidth,keepaspectratio]{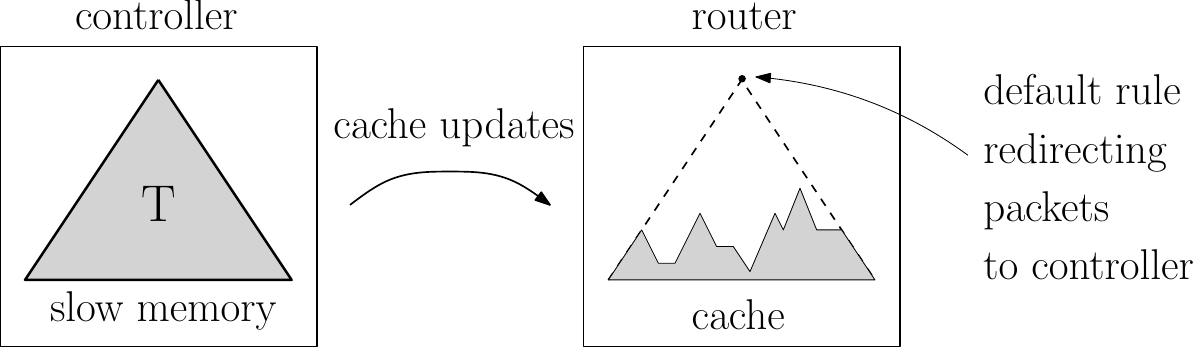}
  \caption{The router (\emph{right}) caches only a subset of all rules, and rules that are not cached are answered by the controller (\emph{left}) that keeps the whole tree of rules. Updates to the rules are passed by the controller to the router.}
  \label{fig:motivation}
\end{figure}

\paragraph{Tree dependencies}

Note that the technical feasibility of this solution heavily depends on the rule dependencies. In the most ubiquitous scenario, the rules are prefixes of IP addresses (they are bit strings). Whenever a packet arrives, the router follows a longest matching prefix (LMP) scheme: it searches for the rules that are prefixes of the destination IP of the packet and among the matching rules it chooses the longest one. In other words, if the prefixes corresponding to rules are stored in the tree\footnote{We do not have to assume that they are actually stored in a real tree; this tree is implicit in the LMP scheme.}, then the tree is traversed from the root downwards, and the last found rule is used. This explains why we require the cached nodes to form a subforest: leaving a less specific rule on the router while evicting a more specific one (i.e., keeping a~tree node in cache while evicting its descendant) will result in a~situation where packets will be forwarded according to the less specific rule, and hence potentially exit through the wrong port. The LMP scheme also ensures that the described approach is implementable: one could simply add an~artificial rule at the tree root in the router (matching an empty prefix). This ensures that when no actual matching rule is found in the router (in the cache), the packet will be forwarded according to this artificial rule to the controller that stores all the rules and can handle all packets appropriately.

So far, the papers on IP rule caching have avoided dependencies either by assuming that rules do not overlap (a~tree has a single level)~\cite{route-caching-flat} or by preprocessing the tree, so that the rules become non-overlapping~\cite{prefix-caching,fib-caching-non-overlapping}. Unfortunately, this could lead to a large inflation of the routing table. A~notable exception is a solution called CacheFlow~\cite{cacheflow}. The CacheFlow model supports dependencies even in the form of directed acyclic graphs. However, CacheFlow was evaluated only experimentally, and no worst-case guarantees were given on the overall cost of caching. Our work provides theoretical foundations for respecting tree dependencies.

\paragraph{Negative requests}

Additionally, a rule may need to be updated. For example, due to a~change communicated by a dynamic routing protocol (e.g., BGP) the action defined by a~rule has to be modified. In such a case, we have to update the rules at the controller: we assume that this cost is zero. (This cost is unavoidable for any algorithm, so such an assumption makes our problem only more difficult.) Furthermore, if the rule is also stored at the router, then we have to pay a~fixed cost of $\alpha$ for updating the router (see the footnote above concerning the cost of fetches and evictions). Such penalties can be easily simulated in our model: we issue a~sequence of $\alpha$ negative requests to the updated node. Such simulation can change the cost by a factor of at most $2$; a simple formal argument is given in \cref{sec:bisimulation}.

\paragraph{Implementability}

Note that the whole input (fed to a tree caching algorithm) is created at the controller: positive requests are caused by cache misses (which redirect packets to the controller) and batches of $\alpha$ negative requests are caused by updates produced by the dynamic routing algorithm run at the controller. Therefore, the whole tree caching algorithm can be implemented in software at the controller alone. Furthermore, our algorithm is a simple counter-based scheme, which can be implemented efficiently and also fine-tuned for speed, see \cref{sec:implementing_counters}.

\paragraph{Other work on forwarding table minimization}

Other approaches for minimizing the number of stored rules were mostly based on \emph{rule compression (aggregation)}, where the set of rules was replaced by another equivalent and smaller set. Optimal aggregation of a fixed routing table can be achieved by dynamic programming~\cite{ortc,fib-compression-two-dimensional}, but the main challenge lies in balancing the achieved compression and the amount of changes to the routing table in the presence of \emph{updates} to this table. While many practical heuristics have been devised by the networking community for this problem~\cite{mms,fib-compression-fifa,fib-compression-globecom10,fib-compression-infocom13,fib-sigcomm,fib-compression-smalta,fib-compression-infocom10}, worst-case analyses were presented only for some restricted scenarios~\cite{fib-icdcs,fib-sirocco,fib-ton}. Combining rule compression and rule caching remains an unexplored area.


\section{Preliminaries}\label{sec:preliminaries}

We denote the height of $T$, i.e., the number of nodes on the longest root-to-leaf path, by $h(T)$. For any non-root node~$v$, we denote its parent by $p(v)$. For any node $v$, $T(v)$ denotes the subtree of $T$ rooted at $v$ (containing~$v$ and all its descendants). A \emph{tree cap} rooted at $v$ is ``an~upper part'' of $T(v)$, i.e., it contains $v$ and if it contains node~$u$, then it also contains all nodes on the path from $u$ to $v$. If $A \subseteq B$ are both tree caps rooted at $v$, then we say that $A$ is a tree cap of $B$. For two disjoint sets $A$ and $B$, we use $A \uplus B$ to denote their union.

We assume discrete time slotted into rounds, with round~$t \geq 1$ corresponding to time interval $(t-1,t)$. In round $t$, the algorithm is given one (positive or negative) request to exactly one tree node and has to process it, i.e., pay associated costs (if any). Right after round~$t$, at time $t$, the algorithm may arbitrarily reorganize its cache, (i)~ensuring that the resulting cache is a subforest of $T$ (i.e., if the cache contains node~$v$, then it contains the entire~$T(v)$) and (ii)~preserving the cache capacity constraint. An algorithm pays $\alpha$ for a~single node fetch or eviction. We denote the contents of the cache at round~$t$ by~$C_t$. (As the cache changes contents only between rounds, $C_t$ is well defined.) We assume that $\alpha$ is an even integer (this assumption may change costs at most by a~constant factor). We assume that the algorithm starts with the empty cache.

We call a non-empty set $X$ a \emph{valid positive changeset} for cache $C$ if $X \cap C = \emptyset$ and $C \cup X$ is a subforest of~$T$, and a~\emph{valid negative changeset} if $X \subseteq C$ and $C \setminus X$ is a subforest of $T$. We call $X$ a~\emph{valid changeset} if it is either a valid positive or a valid negative changeset. Note that the union of positive (negative) changesets is also a valid positive (negative) changeset. We say that the algorithm applies changeset~$X$, if it fetches all nodes from~$X$ (for a positive changeset) and evicts all nodes from $X$ (for a negative one). Note that not all valid changesets may be applied as the algorithm is also limited by its cache capacity ($\kALG$ for an online algorithm and $\kOPT$ for the optimal offline one).


\section{Algorithm}\label{sec:algo}

The algorithm \textsc{Tree Caching} (\ALG) presented in the following is a simple scheme that follows a \emph{rent-or-buy paradigm}: it fetches (or evicts) a changeset $X$ if the cost associated with requests at $X$ reaches the cost of such a fetch or eviction.

More concretely, \ALG operates in multiple phases. The first phase starts at time $0$. \ALG starts each phase with the empty cache and proceeds as follows. Within a phase, every node keeps a counter, which is initially zero. If at round~$t$ \ALG pays~$1$ for serving the request at a node~$v$, it increments the counter of~$v$. Whenever a node is fetched or evicted from the cache, its counter is reset to zero. Note that this implies that the counter of $v$ is equal to the number of negative (positive) requests to~$v$ since its last fetch into the cache (its last eviction from the cache). For a~set $A \subseteq T$, we denote the sum of all counters in $A$ at time $t$ by $\cnt_t(A)$. At time~$t$, \ALG verifies whether there exists a~valid changeset $X$, such that
\begin{itemize}
\item \emph{(saturation property)} $\cnt_t(X) \geq |X| \cdot \alpha$ and
\item \emph{(maximality property)} $\cnt_t(Y) < |Y| \cdot \alpha$ for any valid changeset $Y \supsetneq X$.
\end{itemize}
In this case, the algorithm modifies its cache applying~$X$.

If, at time $t$, \ALG is supposed to fetch some set $X$, but by doing so it would exceed the cache capacity $\kALG$, it evicts all nodes from the cache instead, and starts a~new phase at time~$t$. Such a \emph{final eviction} might not be present in the last phase, in which case we call it \emph{unfinished}.

In \cref{lem:no_over-requested_changesets} (below), we show that at any time, all valid changesets satisfying both properties of \ALG are either all positive or all negative. Furthermore, right after the algorithm applies a changeset, no valid changeset satisfies the saturation property.


\section{Analysis of \ALG}
\label{sec:analysis}

Throughout this section, we fix an input $I$, its partition into phases, and analyze both \ALG and \OPT on a~single fixed phase $P$. We denote the times at which $P$ starts and ends by $\beP$ and $\enP$, respectively, i.e., rounds in~$P$ are numbered from $\beP+1$ to $\enP$. The following technical lemma is proven by induction; the proof is presented in \cref{sec:proof_of_lemma_1}.

\begin{lemma}
\label{lem:no_over-requested_changesets}
Fix any time $t > \beP$. For any valid changeset $X$ for $C_t$, it holds that $\cnt_t(X) \leq |X| \cdot \alpha$. If a~changeset $X$ is applied at time $t$, the following properties hold:
\begin{enumerate}
\item $X$ contains the node requested at round $t$,
\label[property]{lemit:1}
\item $\cnt_t(X) = |X| \cdot \alpha$,
\label[property]{lemit:2}
\item $\cnt_t(Y) < |Y| \cdot \alpha$ for any valid changeset $Y$ for~$C_{t+1}$ (note that $C_{t+1}$ is the cache state right after application of $X$),
\label[property]{lemit:3}
\item $X$ is a tree cap of a tree from $C_{t+1}$ if $X$ is positive and it is a~tree cap of a tree from $C_t$ if $X$ is negative.
\label[property]{lemit:4}
\end{enumerate}
\end{lemma}

In the following, we assume that no positive requests are given to nodes inside the cache and no negative ones to nodes outside of it. (This does not change the behavior of \ALG and can only decrease the cost of \OPT.)

For the sake of analysis, we assume that at time $\enP$, \ALG actually performs a cache fetch (exceeding the cache size limit) and then, at the same time instant, empties the cache. This replacement only increases the cost of \ALG. Let $k_P$ denote the number of nodes in the cache of $\ALG$ at $\enP$. In a~finished phase, we measure it after the artificial fetch, but right before the final eviction, and thus $k_P \geq \kALG + 1$; in an~unfinished (last) phase $k_P \leq \kALG$.

The crucial part of our analysis that culminates in \cref{sec:shifting} is the technique of shifting requests. Namely, we modify the input sequence by shifting requests up or down the tree, so that the resulting input sequence (i) is not harder for \OPT and (ii) is more structured: we may lower bound the cost of \OPT on each node separately and relate it to the cost of \ALG.


\subsection{Event Space and Fields}
\label{sec:event}

In our analysis, we look at a two-dimensional, discrete, spatio-temporal space, called the \emph{event space}. The first dimension is indexed by tree nodes, whose order is an~arbitrary extension of the partial order given by the tree. That is, the parent of a node $v$ is always ``above''~$v$. The second dimension is indexed by round numbers of phase~$P$. The space elements are called \emph{slots}. Some slots are occupied by requests: a~request at node $v$ given at round $t$ occupies slot $(v,t)$. From now on, we will identify $P$ with a set of requests occupying some slots in the event space.

We partition slots of the whole event space into disjoint parts, called \emph{fields}, and we show how this partition is related to the costs of \ALG and \OPT. For any node~$v$ and time $t$, $\last_v(t)$ denotes the last time strictly before~$t$ when node~$v$ changed state from cached to non-cached or vice versa; $\last_v(t) = \beP$ if $v$ did not change its state before $t$ in phase $P$. For a~changeset~$X_t$ applied by \ALG at time $t$, we define the field $F^t$ as
\[
  F^t = \left\{\ (v,r) : v \in X_t \, \wedge\, \last_v(t)+1 \leq r \leq t\ \right\}.
\]
That is, field $F^t$ contains the partial timelines of all the requests that eventually trigger the application of $X_t$ at time~$t$. We say that $F^t$ ends at~$t$. We call field $F^t$ \emph{positive} (\emph{negative}) if $X_t$ is a positive (negative) changeset. An~example of a~partitioning into fields is given in \cref{fig:fields}. We define $\req(F^t)$ as the number of requests belonging to slots of~$F^t$ and let $\size(F^t)$ be the number of involved nodes (note that $\size(F^t) = |X_t|$). The observation below follows immediately from \cref{lem:no_over-requested_changesets}.

\begin{figure}[t]
  \centering
  \paperfigure[width=0.8\columnwidth,keepaspectratio]{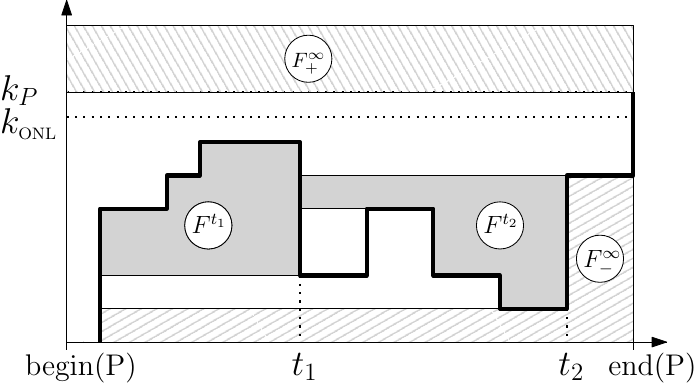}
  \caption{Partitioning of a single phase into fields for a line (a tree with no branches). The thick line represents cache contents. Possible final eviction at $\enP$ is not depicted. $F^{t_1}$ is a~negative field and $F^{t_2}$ is a positive one. In the particular depicted example, nodes are ordered from the leaf (bottom) to the root (top of the picture). We emphasize that for a general, branched tree, some notions (in particular fields) no longer have nice geometric interpretations.}
  \label{fig:fields}
\end{figure}

\begin{observation}
\label{obs:field_requests}
For any field $F$, $\req(F) = \size(F) \cdot \alpha$. All these requests are positive (negative) if $F$ is positive (negative).
\end{observation}

Finally, we call the rest of the event space defined by phase~$P$ the \emph{open field} and denote it by $F^\infty$. The set of all fields except $ F^\infty$ is denoted by $\F$. Let $\size(\F) = \sum_{F \in \F} \size(F)$.

\begin{lemma}
\label{lem:alg_cost}
For any phase $P$ partitioned into a set of fields~$\F \cup \{ F^\infty \}$, it holds that $\ALG(P) \leq 2 \alpha \cdot \size(\F) + \req(F^\infty) + k_P \cdot \alpha$.
\end{lemma}

\begin{IEEEproof}
By \cref{obs:field_requests}, the cost associated with serving the requests from all fields from $\mathcal{F}$ is $\sum_{F \in \F} \alpha \cdot \size(F) = \alpha \cdot \size(\F)$. The cost of the cache reorganization at the fields' ends is exactly the same. The term $\req(F^\infty)$ represents the cost of serving the requests from $F^\infty$ and $k_P \cdot \alpha$ upper-bounds the cost of the final eviction (not present in an unfinished phase).
\end{IEEEproof}


\subsection{Shifting Requests}\label{sec:shifting}

The actual challenge in the proof is to relate the structure of the fields to the cost of \OPT. The rationale behind our construction is based on the following thought experiment. Assume that the phase is unfinished and that the $\size(F) \cdot \alpha$ requests of each field $F \in \F$ are evenly distributed among its nodes, i.e., each node of $F$ receives $\alpha$ of them. Then, the history of any node $v$ alternates between periods spent in positive fields and periods spent in negative fields, each containing exactly $\alpha$ requests. Hence, for any two consecutive periods of a~single node, \OPT has to pay at least $\alpha$: either for the positive requests, or for the negative ones, or for changing the cached/non-cached state of~$v$. This shows that the cost of \OPT can be easily related to $\alpha \cdot \size(\F)$.

Unfortunately, the requests may not be evenly distributed among the nodes. To alleviate this problem, we modify the requests in phase $P$, so that the newly created phase $P'$ is not harder for $\OPT$ and the requests in $P'$ are more evenly distributed. In this construction, the time frame of $P$, the cache states of \ALG and the fields are fixed.

We say that a request placed originally (in phase $P$) at slot~$(v,t)$ is \emph{legally shifted} if its new slot is $(m(v), t)$, where $m(v)$ is $v$ or a descendant of $v$ for a positive request, and $v$ or an~ancestor of $v$ for a negative one. As caches are sub-forests, for any fixed sequence of fetches and evictions within phase~$P$, the associated cost may only decrease when these actions are replayed on the modified requests.

\begin{observation}
\label{obs:pprim_easier_than_p}
If $P'$ is created from $P$ by legally shifting the requests, then $\OPT(P') \leq \OPT(P)$ (for the same initial cache of \OPT).
\end{observation}

The main difficulty lies in keeping the shifted requests within fields. For example, a~negative request from $F$ shifted at round $r$ from node~$u$ to its parent may fall out of $F$, as the parent may still be outside the cache at round~$r$. In effect, a careless (though legal) shifting may lead to a situation where, for a single node~$v$, positive and negative requests no longer interleave, and hence we cannot argue that $\OPT(P')$ is sufficiently large. Below, we show how to legally shift the requests, so that negative requests remain within their fields and are \emph{exactly} evenly distributed among the nodes of these fields, while positive requests remain within positive fields (although not necessarily within their original ones) and are \emph{approximately} evenly distributed.

We extend the notion of $\req(\cdot)$ to arbitrary sets of slots in a natural way. For a set of slots $S$ and a set of nodes~$A$, let $S \cap A = \{(v,r) \in S : v \in A\}$. For a field $F^t$ and a~node~$v \in X_t$, we call $F^t \cap \{v\}$ the \emph{row} of $v$ in $F^t$. We call a row positive (negative) if $F^t$ is positive (negative), and we call it \emph{full} if it contains at least $\alpha/2$ requests.


\subsection{Shifting Negative Requests Up}
\label{sec:negative_shifting}

\begin{lemma}
\label{lem:crucial_lemma_neg}
For any negative field $F^t$, it is possible to legally shift its requests up, so that they remain within $F^t$ and after the modification each node of $X_t$ is requested exactly $\alpha$ times.
\end{lemma}

\begin{IEEEproof}
We proceed in a bottom-up fashion, taking care of a~single node at a time. We maintain a tree cap $Y$ of $X_t$, initially equal to $X_t$, and in each step we remove a leaf $v$ from~$Y$, ensuring that exactly $\alpha$ requests remain at $v$.

For a tree cap $D$ of $Y$ and a round $r \leq t$, let $A_r(D) = \{u \in D : \last_u(t) < r\}$ be the set of nodes of $D$ that have slots in $F^t$ at round $r$, and let $q_r(D)$ be the number of requests in $F^t \cap D$ given at rounds at most $r$. Note that $q_t(D) = \req(F^t \cap D)$. Throughout the process, we maintain the following invariants: (i) $q_t(Y) = |Y| \cdot \alpha$, and (ii) for any tree cap $D$ of $Y$ and any round $r \leq t$, if $A_r(D)$ is a valid negative changeset for $C_r$, then
\[
q_r(D) \leq |A_r(D)| \cdot \alpha.
\]
Before we shift any request, invariant (i) follows by \cref{obs:field_requests} and invariant (ii) by \cref{lem:no_over-requested_changesets}, as $q_r(D) = \cnt_r(A_r(D))$.

If $Y$ consists of a single node, this node has exactly $\alpha$ requests, and we are done. Otherwise, we pick a leaf $v$ of~$Y$. By invariant (ii) applied to $Y \setminus \{v\}$ and round $t$, it holds that $q_t(Y \setminus \{v\}) \leq (|Y|-1) \cdot \alpha$, and thus $v$ has at least~$\alpha$ requests. If $v$ has exactly $\alpha$ requests, we remove $v$ from~$Y$ without shifting anything. Otherwise, we leave the first $\alpha$ requests at~$v$ and shift the remaining ones to $p(v)$.

We now show that the shifted requests stay within $F^t$. Let $s$ be the round of the first shifted request, and assume for a~contradiction that $b = \last_{p(v)}(t) \geq s$. Let $D$ be the path from the root of $Y$ to $v$ (inclusively). The nodes of $D \setminus \{v\}$ are fetched at time $b$ or later, and $A_b(D) = \{v\}$. Thus, $\{v\}$ is a~valid negative changeset for $C_b$, but $q_b(D) \geq \alpha+1$, which contradicts the invariant. Therefore, $b < s$.

It remains to show that the invariant holds for $Y \setminus \{v\}$ after the shift. Clearly,  $q_t(Y \setminus \{v\}) = (|Y|-1) \cdot \alpha$. Now, fix a tree cap $D$ of $Y \setminus \{v\}$ and a round $r \leq t$, such that $A_r(D)$ is a~valid negative changeset for $C_r$. If $p(v) \notin D$ or $r < s$, then the shift does not change $q_r(D)$, and the bound follows as $D$ is also a tree cap of $Y$. Otherwise, $D \cup \{v\}$ is a tree cap of $Y$, and $A_r(D) \cup \{v\}$ is a valid negative changeset for $C_r$. Thus, before the shift, $q_r(D \cup \{v\}) \leq (|A_r(D)|+1) \cdot \alpha$. The shift does not change $q_r(D \cup \{v\})$ and leaves exactly $\alpha$ of these requests at $v$, and hence $q_r(D) \leq |A_r(D)| \cdot \alpha$.
\end{IEEEproof}


\subsection{Shifting Positive Requests Down}
\label{sec:positive_shifting}

We now focus on positive requests. It turns out that a~counterpart of \cref{lem:crucial_lemma_neg} is impossible for positive fields: no shifting strategy that keeps the requests of $F^t$ inside $F^t$ can guarantee even a constant fraction of nodes with $\Omega(\alpha)$ requests. To see this, consider a star with root $v$ and $m$ leaves. First, $v$ receives $m \cdot \alpha - 1$ positive requests. Then, each leaf in turn receives~$\alpha$ positive requests, is fetched, receives $\alpha$ negative requests, and is evicted. Finally, the leaves receive $\alpha + 1$ further positive requests in total, fewer than $\alpha$ each, and \ALG fetches the whole star. In the corresponding field, all requests to $v$ precede the slots of every leaf, and thus no request to~$v$ can be shifted to a leaf within this field.

Still, the average number of requests in a positive field is \emph{at least as large at the top of the field as at its bottom}, in the following sense.

\begin{lemma}
\label{lem:top_heavy}
Fix a positive field $F^t$ and the corresponding changeset $X_t$. For any tree cap $D$ of $X_t$, it holds that $\req(F^t \cap D) \geq |D| \cdot \alpha$.
\end{lemma}

\begin{IEEEproof}
If $D = X_t$, the lemma follows with equality. Otherwise $X_t \setminus D$ is a valid positive changeset for $C_t$. In the latter case, $\req(F^t \cap (X_t \setminus D)) \leq |X_t \setminus D| \cdot \alpha$ by \cref{lem:no_over-requested_changesets}. The lemma follows as $\req(F^t) = |X_t| \cdot \alpha$ by \cref{obs:field_requests}.
\end{IEEEproof}

It would be thus natural to shift the requests of a node~$v \in X_t$ down to the nodes of $T(v) \cap X_t$, in the same manner as we shifted negative requests up. The example above shows that this is impossible within a~single field. Therefore, we shift positive requests across fields: a request to a node~$v$ is moved to a descendant~$u$ of $v$ that is outside the cache at the round of the request. After such a shift, the request belongs to the field that ends with the next fetch of~$u$. The guarantee we obtain is global, i.e., it concerns all positive fields of the phase together. Let $N$ denote the total number of rows of all positive fields from $\F$, i.e., $N = \sum_{F} \size(F)$, where the sum is taken over all positive fields $F \in \F$.

\begin{lemma}
\label{lem:crucial_lemma_pos}
It is possible to legally shift positive requests down, so that each request remains within a positive field (possibly different from its original one) and at least $N / (2 h(T))$ rows of positive fields are full (have at least $\alpha/2$ requests).
\end{lemma}

\begin{IEEEproof}
For a node $v$ and a round $\rho$ such that $v \notin C_\rho$, let $U_\rho(v) = T(v) \setminus C_\rho$. Fix a positive request at slot $(v, \rho)$ of a~field $F^t$ and a node $u \in U_\rho(v)$. As $u \notin C_\rho$ and $u \in T(v) \subseteq C_t \cup X_t$, the first state change of $u$ after round $\rho$ is a fetch at some time $t_u \leq t$. Hence, slot $(u, \rho)$ belongs to a positive field (possibly different from $F^t$). We call the row of $u$ in $F^{t_u}$ the \emph{row of $u$ at round $\rho$}.

We process all positive requests of fields from $\F$ one by one, in an arbitrary order. For a request at slot $(v, \rho)$, we check whether some node $u \in U_\rho(v)$ (possibly $u = v$) has a row at round $\rho$ that contains less than $\alpha/2$ useful requests. If so, we move the request to $(u, \rho)$ and call it \emph{useful}. Otherwise, we keep the request at $(v, \rho)$ and call it \emph{wasted}. A useful request is never moved again. Each row receives at most~$\alpha/2$ useful requests, and hence the number of useful requests is at most $N \cdot \alpha/2$. By \cref{obs:field_requests}, the total number of positive requests in fields from $\F$ is $N \cdot \alpha$, and hence at least $N \cdot \alpha/2$ requests are wasted.

It remains to relate the number of wasted requests to the number $M$ of full rows. We use the following charging scheme. Fix a row $R$ of a node $v$ in a positive field $F^t$ that contains $W \geq 1$ wasted requests, and let $\rho^*$ be the round of the last of them. Then $R$ charges $\alpha$ to the row at round $\rho^*$ of every node from $U_{\rho^*}(v)$. This scheme has the following properties.
\begin{itemize}
\item The total charge distributed by $R$ is $|U_{\rho^*}(v)| \cdot \alpha \geq W$. 
Indeed, as these $W$ wasted requests are all original requests to $v$, and as $U_{\rho^*}(v)$ is a valid positive changeset for~$C_{\rho^*}$, \cref{lem:no_over-requested_changesets} implies that $W \leq \cnt_{\rho^*}(U_{\rho^*}(v)) \leq |U_{\rho^*}(v)| \cdot \alpha$.
\item Only full rows are charged. Indeed, the request at $(v, \rho^*)$ was declared wasted, and thus at the time it was processed, all rows charged by $R$ already contained $\alpha/2$ useful requests.
\item Every row $R'$ of a node $u$ receives a total charge of at most $h(T) \cdot \alpha$. Indeed, $R'$ may be charged by a row of a node $v$ only if $u \in T(v)$, i.e., only if $v$ is one of the at most $h(T)$ nodes on the path from $u$ to the root. Fix such a node $v$. Within all rounds of $R'$, $u$ is outside the cache, and thus $v$ is outside the cache as well. A row of~$v$ that charges $R'$ shares round $\rho^*$ with $R'$, and hence $R'$ is charged by at most one row of $v$.
\end{itemize}
By the first property, the total charge is at least the number of wasted requests, i.e., at least $N \cdot \alpha/2$. By the other two properties, the total charge is at most $h(T) \cdot \alpha \cdot M$. Thus, $h(T) \cdot \alpha \cdot M \geq N \cdot (\alpha/2)$ and therefore $M \geq N / (2 h(T))$.
\end{IEEEproof}


\subsection{Using Request Shifting for Bounding OPT}
\label{sec:lower-bound}

Finally, we use our request shifting to relate $\size(\F)$ to the cost of $\OPT$ in a single phase $P$. Recall that $k_P$ denotes the size of \ALG's cache at the end of $P$. We assume that \OPT may start the phase with an arbitrary state of the cache.

\begin{lemma}
\label{lem:leftovers}
For any phase $P$, it holds that $\OPT(P) \geq (\size(\F) / (4 h(T)) - k_P) \cdot \alpha/2$.
\end{lemma}

\begin{IEEEproof}
We transform $P$ into a phase $P'$ using the legal shifts of \cref{lem:crucial_lemma_neg,lem:crucial_lemma_pos}. By \cref{obs:pprim_easier_than_p}, it is sufficient to show that $\OPT(P') \geq (\size(\F) / (4 h(T)) - k_P) \cdot \alpha/2$.

\begin{figure}[t]
\centering
\paperfigure[width=0.9\columnwidth,keepaspectratio]{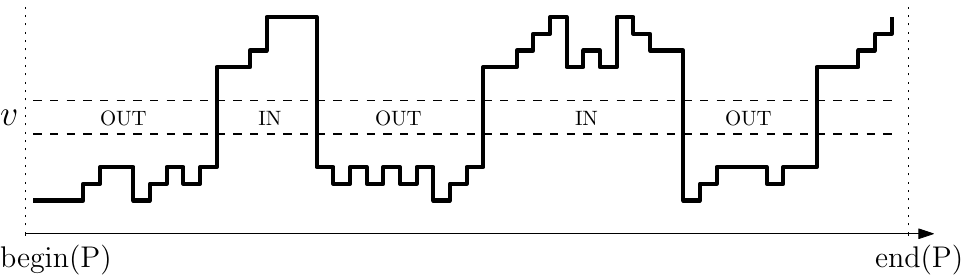}
\caption{Partitioning of the phase into interleaving \pin and \pout periods for node $v$. The thick line represents cache contents. The \emph{leftover} \pout period (the last one) is present for node $v$ as it has finished phase $P$ inside \ALG's cache. The periods can be followed by requests contained in $F^\infty$.}
\label{fig:leftover}
\end{figure}

We focus on a single node $v$. As \ALG starts each phase with an empty cache, the positive and negative rows of $v$ alternate, starting with a positive one (if $v$ has any rows). We call them \pout \emph{periods} and \pin \emph{periods}, respectively. A final (possibly empty) part of the history of $v$, when $v$ is in the open field $F^\infty$, does not belong to any period. An~example is given in~\cref{fig:leftover}.

The total number of \pout periods of all nodes is denoted $N$ as in the proof of \cref{lem:crucial_lemma_pos}, and the total number of all periods is $\size(\F)$. For the $k_P$ nodes that are in \ALG's cache at the end of the phase (right before the final eviction, if any), and only for them, the history ends with an \pout period not followed by an \pin period. We call these periods \emph{leftover periods}. Thus, the number of \pin periods is $N - k_P$, and hence $N \geq \size(\F)/2$.

By \cref{lem:crucial_lemma_neg}, all \pin periods contain exactly $\alpha$ requests, and by \cref{lem:crucial_lemma_pos}, at least $N/(2 h(T))$ \pout periods are full. Thus, there are at least $N/(2 h(T)) - k_P$ full non-leftover \pout periods; each of them together with the following \pin period constitutes a \emph{full \pout-\pin pair}. \OPT has to pay at least $\alpha/2$ for the node in the course of the history described by a~full \pout-\pin pair: it pays $\alpha$ either for changing the cached/non-cached state of the node, or $\alpha/2$ for all positive requests, or $\alpha$ for all negative ones. As these pairs are disjoint, $\OPT(P') \geq (N/(2 h(T)) - k_P) \cdot \alpha/2 \geq (\size(\F) / (4 h(T)) - k_P) \cdot \alpha/2$.
\end{IEEEproof}


\subsection{Competitive Ratio}
\label{sec:comp_ratio}

To relate the cost of \OPT to \ALG in a single phase $P$, we still need to upper-bound $\req (F^\infty)$ and relate $k_P \cdot \alpha$ to the cost of $\OPT$ (i.e., compare the bounds on \ALG and \OPT provided by \cref{lem:alg_cost,lem:leftovers}, respectively).

For the next two lemmas, we define $\VOPT$ as the set of all nodes that were in \OPT's cache at some time of~$P$ and let $\VOPTC = T \setminus \VOPT$. Note that $\VOPT$ is a union of sub-forests (consisting of nodes present in \OPT's cache at consecutive times), and hence a subforest itself.

\begin{lemma}
\label{lem:f_infty}
For any phase $P$, it holds that $\req (F^\infty) \leq 2 \cdot \kALG \cdot \alpha + 2 \cdot \OPT(P)$.
\end{lemma}

\begin{IEEEproof}
We assume first that $P$ is a finished phase. Then, $P$ ends with an artificial fetch of $X_{\enP}$ at time $\enP$ (followed by the final eviction). We split $F^\infty$ into two disjoint parts (see \cref{fig:fields}):
\begin{align*}
  F^\infty_- = &\; \{(v, t): v \in C_{\enP}, t > \last_v(\enP)\}, \\
  F^\infty_+ = &\; \{(v, t): v \notin C_{\enP} \uplus X_{\enP}, \,
           t > \last_v(\enP)\}.
\end{align*}
Note that $F^\infty_-$ contains only negative requests and $F^\infty_+$ only positive ones. As $\req(F^\infty) = \req(F^\infty_-) + \req (F^\infty_+ \cap \VOPTC) + \req (F^\infty_+ \cap \VOPT)$, we estimate each of these summands separately.
\begin{itemize}
\item The counters of the nodes of $C_{\enP}$ at time $\enP$ count the requests of $F^\infty_-$. As $C_{\enP}$ is a valid negative changeset for $C_{\enP}$, \cref{lem:no_over-requested_changesets} implies $\req(F^{\infty}_-) = \cnt_{\enP}(C_{\enP}) \leq |C_{\enP}| \cdot \alpha \leq \kALG \cdot \alpha$.
\item All the (positive) requests at nodes of $\VOPTC$ are paid by \OPT, and hence $\req(F^\infty_+ \cap \VOPTC) \leq \OPT(P)$.
\item $F^\infty_+$ is a valid changeset for cache $C_{\enP} \uplus X_{\enP}$. As $\VOPT$ is a subforest of $T$, $F^\infty_+ \cap \VOPT$ is also a valid changeset for the cache $C_{\enP} \uplus X_{\enP}$. Therefore, $\req (F^\infty_+ \cap \VOPT) \leq \size(F^\infty_+ \cap \VOPT) \cdot \alpha$, as otherwise the set fetched at time $\enP$ would not be maximal. (\ALG could then fetch $X_{\enP} \uplus (F^\infty_+ \cap \VOPT)$ instead of $X_{\enP}$.) Thus, $\req (F^\infty_+ \cap \VOPT) \leq |\VOPT| \cdot \alpha = \kOPT \cdot \alpha + (|\VOPT| - \kOPT) \cdot \alpha \leq \kALG \cdot \alpha + \OPT(P)$. The last inequality follows as, independently of the initial state, \OPT needs to fetch at least $|\VOPT| - \kOPT$ nodes to the cache during $P$.
\end{itemize}
Hence, in total, $\req (F^\infty) \leq 2 \cdot \kALG \cdot \alpha + 2 \cdot \OPT(P)$ for a~finished phase $P$.

We note that if there was no cache change at $\enP$, the analysis above would hold, with virtually no change, for $X_{\enP} = \emptyset$. Therefore, for an unfinished phase $P$ ending with a fetch or ending without cache change at $\enP$, the bound on $\req(F^\infty)$ still holds. However, if an unfinished phase~$P$ ends with an eviction, then we look at the last eviction-free time $\tau$ of~$P$. We now observe the evolution of field $F^\infty$ from time~$\tau$ till $\enP$. At time $\tau$, $\req(F^\infty) \leq 2 \cdot \kALG \cdot \alpha + 2 \cdot \OPT(P)$, as proven above. Furthermore, in subsequent times, it may only decrease: at any round, $F^\infty$ gets at most one additional request, but on eviction $\req(F^\infty)$ decreases by $\alpha$ times the number of evicted nodes (i.e., at least by $\alpha \geq 1$). Hence, the value of $\req(F^\infty)$ at $\enP$ is also at most $2 \cdot \kALG \cdot \alpha + 2 \cdot \OPT(P)$.
\end{IEEEproof}

By combining \cref{lem:alg_cost,lem:leftovers,lem:f_infty}, we immediately obtain the following corollary (holding for both finished and unfinished phases).

\begin{corollary}
\label{cor:any_phase_bound}
For any phase $P$, it holds that $\ALG(P) \leq O(h(T)) \cdot \OPT(P) + O(h(T) \cdot (k_P + \kALG) \cdot \alpha)$.
\end{corollary}

Using the corollary above, it remains to bound the value of~$k_P$. This is easy for an unfinished phase, as $k_P \leq \kALG$ there. For a~finished phase, we provide another bound.

\begin{lemma}
\label{lem:opt_bound2}
For any finished phase $P$, it holds that $k_P \cdot \alpha \leq \OPT(P) \cdot (\kALG + 1) / (\kALG + 1 - \kOPT)$.
\end{lemma}

\begin{IEEEproof}
First, we compute the number of positive requests in $\VOPTC$. Let $X_{t_1}, X_{t_2}, \ldots, X_{t_s}$ be all positive changesets applied by \ALG in~$P$. For any~$t$, let $X'_t = X_t \setminus \VOPT$. As $X_t$ is some tree cap and $\VOPT$ is a~subforest of $T$, $X'_t$ is a~tree cap of $X_t$. By \cref{lem:top_heavy}, the number of requests to nodes of~$X'_t$ in field $F^t$ is at least $|X'_t| \cdot \alpha$. These requests for different changesets $X_t$ are disjoint and they are all outside of $\VOPT$. Hence the total number of positive requests outside of $\VOPT$ is at least $\sum_{i=1}^s |X'_{t_i}| \cdot \alpha$, where $\sum_{i=1}^s |X'_{t_i}| \geq |\bigcup_{i=1}^s X'_{t_i}| = |(\bigcup_{i=1}^s X_{t_i}) \setminus \VOPT| \geq |\bigcup_{i=1}^s X_{t_i}| - |\VOPT| \geq k_P - |\VOPT|$.

Now $\OPT(P)$ can be split into the cost associated with nodes from $\VOPT$ and $\VOPTC$, respectively. For the former part, \OPT has to pay at least $(|\VOPT| - \kOPT) \cdot \alpha$ for the fetches alone. For the latter part, it has to pay $1$ for each of at least $(k_P - |\VOPT|) \cdot \alpha$ positive requests outside of $\VOPT$. Hence, $\OPT(P) \geq (|\VOPT| - \kOPT) \cdot \alpha + (k_P - |\VOPT|) \cdot \alpha = (k_P - \kOPT) \cdot \alpha$. Then, $k_P \cdot \alpha \leq k_P \cdot \OPT(P) / (k_P - \kOPT)$. As the phase is finished, $k_P \geq \kALG + 1$, and thus $k_P \cdot \alpha \leq (\kALG + 1) \cdot \OPT(P) / (\kALG + 1 - \kOPT)$.
\end{IEEEproof}

\begin{theorem}
The algorithm \ALG is $O(h(T) \cdot \kALG/(\kALG-\kOPT+1))$-competitive.
\end{theorem}

\begin{IEEEproof}
Let $R = \kALG/(\kALG-\kOPT+1)$. We split an input~$I$ into a sequence of finished phases followed by a single unfinished phase (which may not be present). For a~finished phase $P$, we have $k_P > \kALG$, and hence \cref{cor:any_phase_bound,lem:opt_bound2} imply that $\ALG(P) \leq O(R \cdot h(T)) \cdot \OPT(P)$. For an unfinished phase, $k_P \leq \kALG$, and therefore, by \cref{cor:any_phase_bound}, $\ALG(P) \leq O(h(T)) \cdot \OPT(P) + O(h(T) \cdot \kALG \cdot \alpha)$. Summing over all phases of $I$ yields $\ALG(I) \leq O(R \cdot h(T)) \cdot \OPT(I) + O(h(T) \cdot \kALG \cdot \alpha)$.
\end{IEEEproof}


\section{Implementation of \ALG}\label{sec:implementing_counters}

Recall that at each time $t$, \ALG verifies the existence of a valid changeset that satisfies saturation and maximality properties (see the definition of \ALG in \cref{sec:algo}). Here, we show that this operation can be performed efficiently. In particular, in the next two subsections, we prove the following theorem.

\begin{theorem}
\ALG can be implemented using $O(|T|)$ additional memory, so that to make a decision at time~$t$, it performs $O(h(T) + \max \{ h(T), \degree(T) \} \cdot |X_t|)$ operations, where $\degree(T)$ is the maximum node degree in $T$ and $X_t$ is the changeset applied at time $t$ ($|X_t| = 0$ if no changeset is applied).
\end{theorem}

At the start of a phase, \ALG has to reset all counters and stored values. To avoid the cost $O(|T|)$ of this operation, we tag each of them with the index of the phase in which it was last modified, and we treat an entry with an outdated tag as having its initial value.

Let $v_t$ be the node requested at round $t$. Note that we may restrict our attention to requests that entail a~cost for \ALG, as otherwise its counters remain unchanged and certainly \ALG does not change cache contents. We use \cref{lem:no_over-requested_changesets} to restrict possible candidates for changesets that can be applied at time~$t$. First, we note that if a~node $v_t$ requested at round $t$ is outside the cache, then, at time~$t$, \ALG may only fetch some changeset, and otherwise it may only evict some changeset. Therefore, we may construct two separate schemes, one governing fetches and one for evictions.

In \cref{sec:implementing_positive_counters}, using \cref{lem:no_over-requested_changesets}, we show that after processing a~positive request, \ALG needs to verify at most~$h(T)$ possible positive changesets, each in constant time, using an~auxiliary data structure. The cost of updating this structure at time $t$ is $O(h(T) + h(T) \cdot |X_t|)$.

The situation for negative changesets is more complex as even after applying \cref{lem:no_over-requested_changesets} there are still exponentially many valid negative changesets to consider. In \cref{sec:implementing_negative_counters}, we construct an~auxiliary data structure that returns a viable candidate in time $O(h(T) + \degree(T) \cdot |X_t|)$. The update of this structure at time $t$ can be also done in $O(h(T) + \degree(T) \cdot |X_t|)$ operations.

\subsection{Positive Requests and Fetches}
\label{sec:implementing_positive_counters}

At any time $t$ and for any non-cached node $u$, we may define $P_t(u)$ as a tree cap rooted at $u$ containing all non-cached nodes from $T(u)$. During an execution of \ALG, we maintain two values for each non-cached node~$u$: $\cnt_t(P_t(u))$ and $|P_t(u)|$. When a counter at node~$v_t$ is incremented, we update $\cnt_t(P_t(u))$ for each ancestor~$u$ of~$v_t$ (at most $h(T)$ updated values). Furthermore, if a node~$v$ changes its state from cached to non-cached (or vice versa), we update the values of $|P_t(u)|$ and $\cnt_t(P_t(u))$ for any ancestor $u$ of $v$ (at most $h(T)$ updates per node that changes its state). Therefore, the total cost of updating these structures at time $t$ is at most $O(h(T) + h(T) \cdot |X_t|)$.

By \cref{lem:no_over-requested_changesets}, a positive valid changeset fetched at time~$t$ has to contain $v_t$ and is a single tree cap. Such a~tree cap has to be equal to $P_t(u)$ for some ancestor~$u$ of $v_t$ (possibly $u = v_t$). Hence, we may iterate over all ancestors $u$ of $v_t$, starting from the tree root and ending at $v_t$, and we stop at the first node~$u$, for which $P_t(u)$ is saturated (i.e., $\cnt_t(P_t(u)) \geq |P_t(u)| \cdot \alpha$). If such a $u$ is found, the corresponding set $P_t(u)$ satisfies also the maximality condition (cf.~the definition of \ALG). Indeed, any valid changeset $Y \supsetneq P_t(u)$ is a disjoint union of $P_t(u')$ for some ancestor $u'$ of $u$ (possibly $u' = u$) and of a~set $Z$ of tree caps not containing $v_t$. If $u' \neq u$, then $P_t(u')$ was already verified to be non-saturated. Furthermore, $\cnt_t(Z) < |Z| \cdot \alpha$ for a non-empty $Z$ by \cref{item:prop2,item:propmid} of \cref{cla:inductive_properties}. Hence, $\cnt_t(Y) < |Y| \cdot \alpha$. Therefore, in such a~case, \ALG fetches~$P_t(u)$. Otherwise, if no saturated changeset is found, \ALG does nothing. Checking all ancestors of $v_t$ can be performed in time~$O(h(T))$.

\subsection{Negative Requests and Evictions}
\label{sec:implementing_negative_counters}

Handling evictions is more complex. If the request to node~$v_t$ at round $t$ was negative, \cref{lem:no_over-requested_changesets} tells us only that the negative changeset evicted by \ALG has to be a tree cap rooted at $u$, where $u$ is the root of the cached tree containing $v_t$. There are exponentially many such tree caps, and hence their naïve verification is intractable. To alleviate this problem, we introduce the following helper notion. For any set of cached nodes~$A$ and any time $t$, let
\[
  \val_t(A) = \cnt_t(A) - |A| \cdot \alpha + \frac{|A|}{|T|+1}.
\]
Note that for any non-empty set $A$, $\val_t(A) \neq 0$ as the first two terms are integers and $|A|/(|T|+1) \in (0,1)$. Furthermore, $\val_t$ is additive: for two disjoint sets $A$ and $B$, $\val_t(A \uplus B) = \val_t(A) + \val_t(B)$. For any time~$t$ and a cached node $u$, we define
\begin{align*}
  H_t(u) = \arg \max_D \{ \val_t(D) : &\; \textnormal{$D$ is a non-empty tree cap} \\
    & \quad \textnormal{rooted at $u$} \}.
\end{align*}
Our scheme maintains the value $H_t(u)$ for any cached node $u$. To this end, we observe that $H_t(u)$ can be defined recursively as follows. Let $H'_t(u) = H_t(u)$ if $\val_t(H_t(u)) > 0$ and $H'_t(u) = \emptyset$ otherwise. Then, for any cached node $u$ and time~$t$, by the additivity of $\val_t$,
\begin{equation*}
  H_t(u) = \{ u \} \; \uplus \biguplus_\textnormal{$w$ is a child of $u$} H'_t(w).
\end{equation*}
Each cached node $u$ keeps the value $\val_t(H_t(u))$. Note that set $H_t(u)$ itself can be recovered from this information: we iterate over all children of $u$ (at most $\degree(T)$ of them) and for each child $w$, if $\val_t(H_t(w)) > 0$, we recursively compute set $H_t(w)$. Thus, the total time for constructing $H_t(u)$ is $O(\degree(T) \cdot |H_t(u)|)$.

During an execution of \ALG, we update stored values accordingly. That is, whenever a~counter at a cached node $v_t$ is incremented, we update $\val_t(H_t(u))$ values for each cached ancestor $u$ of $v_t$, starting from \mbox{$u = v_t$} and proceeding towards the cached tree root. Any such update can be performed in constant time, and the total time is thus $O(h(T))$. For a~cache change, we process nodes from the changeset iteratively, starting with nodes closest to the root in case of an~eviction and furthest from the root in case of a fetch. For any such node $u$, we appropriately stop or start maintaining the corresponding value of $\val_t(H_t(u))$. The latter requires looking up the stored values at all its children. As $u$ does not have cached ancestors, sets $H_t$ (and hence also the stored values) at other nodes remain unchanged. In total, the cost of updating all $H_t$ values at time $t$ is at most $O(h(T) + \degree(T) \cdot |X_t|)$.

Finally, we show how to use sets $H_t$ to quickly choose a~valid changeset for eviction. Recall that for a negative request~$v_t$, the changeset to be evicted has to be a tree cap rooted at $u$, where $u$ is the root of a cached subtree containing~$v_t$. For succinctness, we use $H^u$ to denote $H_t(u)$. We show that if $\val_t(H^u) < 0$, then there is no valid negative changeset that is saturated, and hence \ALG does not perform any action, and if $\val_t(H^u) > 0$, then $H^u$ is both saturated and maximal, and hence \ALG may evict~$H^u$.

\begin{enumerate}
\item First, assume that $\val_t(H^u) < 0$. Then, for any tree cap~$X$ rooted at~$u$, it holds that $\cnt_t(X) - |X| \cdot \alpha < \val_t(X) \leq \val_t(H^u) < 0$, i.e., $X$ is not saturated, and hence cannot be evicted by \ALG.

\item Second, assume that $\val_t(H^u) > 0$. As $\cnt_t(H^u) - |H^u| \cdot \alpha$ is an integer and $|H^u|/(|T|+1) < 1$, it holds that $\cnt_t(H^u) - |H^u| \cdot \alpha \geq 0$, i.e., $H^u$ is saturated. Moreover, by \cref{lem:no_over-requested_changesets}, $\cnt_t(H^u) \leq |H^u| \cdot \alpha$, and therefore $\cnt_t(H^u) - |H^u| \cdot \alpha = 0$, i.e., $\val_t(H^u) = |H^u| / (|T|+1)$. It remains to show that $H^u$ is maximal, i.e., there is no valid saturated changeset $Y \supsetneq H^u$. By \cref{lem:no_over-requested_changesets}, $Y$ has to be a tree cap rooted at $u$ as well. If $Y$ was saturated, $\val_t(Y) = \cnt_t(Y) - |Y| \cdot \alpha + |Y| / (|T|+1) \geq |Y| / (|T|+1) > |H^u|/(|T|+1) = \val_t(H^u)$, which would contradict the definition of~$H^u$.
\end{enumerate}

Note that node $u$ can be found in time $O(h(T))$, and the actual set~$H^u$ (of size $|X_t|$) can be computed in time $O(\degree(T) \cdot |X_t|)$. Therefore, the total time for finding the set~$X_t$ is $O(h(T) + \degree(T) \cdot |X_t|)$.

\section{Simulator and Methodology}
\label{sec:experimental_implementation}

In this section, we describe the simulator that we implemented for our empirical study of tree caching. In the simulator, all algorithms are run on the same rule tree and the same packet stream, with the same cache capacity and the same accounting of update costs. We focus on caching decisions and abstract away from switch-specific details, such as the TCAM layout or the encoding of update messages.

\subsection{Simulator Structure}

The simulator is a standalone C++ program. It first loads a set of IPv4 prefix rules and builds the forest of prefix containment that is implicit in the LMP scheme (cf.~\cref{sec:motivation}). Each node corresponds to a prefix rule, and parent-child edges are defined by prefix containment. We attach the roots of this forest to a virtual root that corresponds to the empty prefix. The virtual root is used for lookups and for building the tree, but it is not a cacheable rule. That is, the virtual root is never cached, and the caching algorithms operate on the subtrees below it. A packet is represented by its destination IPv4 address. The simulator maps each packet to its longest matching rule and passes a request to this rule to the evaluated algorithm.

For \ALG and the dependency-aware online baselines, caching a rule requires caching all more specific rules from its subtree. Hence, the cached rules form a subforest, as in our formal model. The batch heuristics that use cover rules represent dependencies differently: a rule may be installed together with artificial \emph{cover rules}, which forward the packets of the affected subregions to the controller. For such heuristics, the set of cached original rules need not be a subforest, but the cover rules ensure that every packet is still classified correctly.

For each run, the simulator reports the number of packets, hits and misses, the numbers of fetched and evicted rules, the number of reconfigurations, the maximum number of rules installed or removed in a single reconfiguration, and the running time. A packet is a \emph{hit} if its longest matching rule is cached when the packet arrives. All other packets are \emph{misses}. These include ordinary cache misses (even if the rule is fetched immediately after serving the request), requests that an algorithm bypasses or cannot install because their dependencies exceed the cache capacity, and packets that match only the virtual root. Thus, every packet contributes exactly once to either hits or misses, and $\textnormal{misses}=\textnormal{packets}-\textnormal{hits}$. The evictions at the ends of \ALG's phases are counted separately from the evictions caused by negative requests; the latter counter is always zero in our experiments, as explained next.

\subsection{Algorithm Implementations}
\label{sec:algorithms_implemented}

The algorithms implemented in the simulator are listed in \cref{tab:implemented_algorithms}. For \ALG, we implemented the counters and the fetching procedure of \cref{sec:implementing_positive_counters}, together with the final eviction of the whole cache when fetching a saturated tree cap would exceed the cache capacity. The simulator processes packet arrivals only. There are no rule updates, and hence no negative requests: the counters of cached nodes are never incremented and the eviction procedure of \cref{sec:implementing_negative_counters} is never invoked. That is, our experiments evaluate \ALG on inputs consisting of positive requests only.

\begin{table}[t]
  \centering
  \caption{Algorithms implemented in the simulator.}
  \label{tab:implemented_algorithms}
  \begin{tabular}{@{}p{0.28\columnwidth}p{0.14\columnwidth}p{0.48\columnwidth}@{}}
    \toprule
    Algorithm & Type & Purpose \\
    \midrule
    \ALG (positive path) & Online & Tree-cap online algorithm \\
    \textsc{Bucketing} & Online & Dependency-aware baseline \\
    \textsc{BucketingBypass} & Online & Bypassing baseline \\
    \textsc{LRU} & Online & Hit-rate reference \\
    \textsc{DependentSet} & Batch & CacheFlow-inspired dependent set \\
    \textsc{CoverSet} & Batch & CacheFlow-inspired cover set \\
    \textsc{MixedSet} & Batch & CacheFlow-inspired mixed set \\
    \bottomrule
  \end{tabular}
\end{table}

\textsc{Bucketing} and \textsc{BucketingBypass} are the randomized online algorithms of Dallot et al.~\cite{dependency-aware-caching}. They were designed for dependencies forming arbitrary directed acyclic graphs, and we specialize them to our rule trees, where the dependencies of a rule are its descendants. \textsc{Bucketing} maintains randomized eviction buckets. On a miss, it fetches the requested rule together with its uncached dependencies and, whenever space is needed, evicts a~feasible rule from the bucket frontier. \textsc{BucketingBypass} first examines the bucket mass below the first uncached dependency on the path from the root to the requested rule. If this mass exceeds the threshold of the algorithm, the request is served by the controller and the rule is not fetched immediately. Otherwise, \textsc{BucketingBypass} uses the same randomized eviction rule as \textsc{Bucketing}. Bypassing avoids some potentially large fetches, but both algorithms may still update rules frequently.

The model of Dallot et al.\ assumes that each requested item fits in the cache together with its dependencies. In our traffic-derived rule tree (described in \cref{sec:traces}), five rules have subtrees of more than 754 rules, i.e., larger than the cache in our main experiments; the largest such subtree has 1,836 rules. Hence, the assumption is violated for requests to these rules. For such a request, both implementations leave the cache unchanged and count a miss. We use \textsc{BucketingBypass} as the main baseline, as bypassing is a part of its model, whereas \textsc{Bucketing} has to fetch on every miss.

The CacheFlow-style algorithms are batch heuristics that follow the rule selection ideas of CacheFlow~\cite{cacheflow}, specialized to a prefix tree. The input is divided into \emph{epochs} of a fixed number of packets. During an epoch, such an algorithm counts the packets matching each rule. At the end of the epoch, it computes a new cache and installs the difference between the old and the new cache. Each batch algorithm starts with an~empty cache, and its first installation happens only at the end of the first epoch. (In particular, if the epoch is longer than the trace, nothing is ever installed.) Hence, the epoch length is a crucial parameter: short epochs make the algorithm more responsive but increase the number of updates, while long epochs decrease the number of updates but make the cache stale.

We implemented three variants, which differ in the way they handle dependencies. \textsc{DependentSet} selects whole subtrees, choosing them by the ratio of their traffic to their size. \textsc{CoverSet} may instead install a rule together with cover rules of higher priority. Recall that a cover rule is a TCAM entry that forwards the covered traffic to the controller, and thus it preserves the correctness of classification. \textsc{MixedSet} greedily chooses between these two options at each selection step.

We model the reconfiguration at the end of an epoch as a~single atomic step. In this step, an entry present in both the old and the new cache is counted once, each insertion and each deletion counts once, and switching a rule between its normal and its cover representation counts as one deletion and one insertion. We do not implement the remaining parts of CacheFlow, such as the maintenance of the dependency graph or the protocol that orders TCAM updates. For this reason, we call these algorithms CacheFlow-style.

Finally, \textsc{LRU} ignores the subforest constraint, and hence it cannot be used for caching forwarding rules. We include it only as a reference point for the hit rate: it shows what hit rate a replacement policy can achieve if dependencies are ignored.

\subsection{Traces and Rule Trees}
\label{sec:traces}

We performed extensive experiments using both synthetic and real-world traces.

To generate a synthetic trace, we assign Zipf ranks to the rules according to a random permutation (fixed by a seed), so that the popularity of a rule is not correlated with its depth. To generate a packet for a chosen rule, we sample its destination address from the part of the prefix that is not covered by the children of the rule. Thus, the longest matching rule of the packet is the chosen rule. A rule whose prefix is entirely covered by its children cannot be requested this way. In the traffic-derived rule tree (described below), 7,390 out of 7,542 rules can be requested.

We use two types of synthetic traces. In \emph{stationary} traces, all packets are drawn from the same Zipf distribution; we use them to control the skew. \emph{Phase-shifted} traces serve as stress tests in which the popular part of the tree changes over time. Such a trace is divided into phases (unrelated to the phases of \ALG). In each phase, 60\% of the packets are drawn from a single subtree rooted at depth~$2$, and the remaining ones from the global distribution. The tree has 273 eligible subtrees rooted at depth~$2$. We shuffle them and use at most 125 phases, so that no subtree is popular twice within a trace. Short phases favor online algorithms, while long phases give batch algorithms more time to adapt to a stable distribution. We note that such abrupt changes are intentional. We align phase boundaries with epoch boundaries: an epoch either divides a~phase or spans an integral number of phases. This way, no epoch partially mixes two phases, which is advantageous for the batch algorithms.

For the experiments on real traffic, we use 16 anonymized packet captures from the Equinix NYC backbone monitor of the Center for Applied Internet Data Analysis (CAIDA)~\cite{caida-nyc-trace} (direction A), collected on January 17, 2019, between 12:59:10 and 13:14:59 UTC. CAIDA anonymizes addresses with the prefix-preserving Crypto-PAn transformation, and hence prefix containment is preserved in the anonymized address space. From these captures, we derive a rule tree, which we call the \emph{traffic-derived rule tree}. It contains prefixes of lengths 8, 12, 16, 20, and 24: a prefix of length 8, 12, 16, 20, or 24 is retained if at least 150,000, 30,000, 15,000, 7,500, or 3,000 packets, respectively, match it. This way, the tree contains the prefixes that carry a significant part of the traffic, together with the containment structure needed for LMP lookups.

The resulting CAIDA trace has 589,630,714 packets. The traffic-derived rule tree has 7,542 rules and height~$5$. It contains 33 prefixes of length~$8$, 127 of length~$12$, 300 of length~$16$, 956 of length~$20$, and 6,126 of length~$24$. An ordinary least-squares fit to the rank-frequency curve on the log--log scale gives the Zipf exponent of approximately 1.47, which we also use for synthetic traces. To measure how quickly the popular rules change, we split the trace into windows of 500,000 packets and compare consecutive windows. The Jaccard distance between their sets of 754 most frequent rules is approximately 0.253, i.e., on average 109 of these 754 rules are replaced between consecutive windows. The Jensen--Shannon divergence between their frequency distributions is approximately 0.044. Thus, the set of popular rules changes moderately, and abrupt changes like those in our phase-shifted traces do not occur.

Out of these 589,630,714 packets, 589,615,105 match a~retained rule. The remaining 15,609 packets match only the virtual root, and hence they are misses for every algorithm.

The traffic-derived rule tree preserves the packet order and the locality of a real backbone trace. However, it differs from the forwarding table of a real router. First, the same trace determines both the retained prefixes and the request frequencies. Second, the height of the tree results from our choice of prefix lengths. Third, its 7,542 rules are far fewer than the rules of a full routing table. Fourth, the captures contain no next hops. In a real table, a rule with the same next hop as its parent imposes no dependency, as the parent forwards its packets correctly. Without next hops, we treat every containment relation as a dependency.

To summarize, the real trace is highly skewed and relatively stable at the level of prefix rules. Stationary Zipf traces are its natural synthetic counterpart, while phase-shifted traces test how well the algorithms adapt.

\subsection{Cost Accounting}

The simulator counts updates in two ways. First, it counts individual fetched and evicted rules, as in our formal model. Second, it counts \emph{update events}: a tree cap fetched or evicted by \ALG is a single event, and at the end of an epoch, a batch algorithm performs at most two events, one for all installed rules and one for all removed ones. We use the event counts only for diagnostic purposes. Were they used as the cost measure, a batch algorithm could reduce its update cost to zero simply by making its epochs longer. All costs reported below are computed according to the formal model, i.e.,
\[
  \textnormal{cost}
    =
  \textnormal{misses}
    +
  \alpha \cdot
  (\textnormal{rules fetched} + \textnormal{rules evicted}) .
\]
That is, each changed TCAM entry counts as one update, no matter whether several updates are sent to the router in a single message. As mentioned above, replacing a rule by a cover rule or vice versa counts as one removal and one installation. Furthermore, a cover rule is charged like any other rule. This way, the node accounting of our model extends to all TCAM entries.

Note that for \ALG, the parameter $\alpha$ plays two roles: it is the cost of an update in the formula above and the saturation threshold used by the algorithm. The other algorithms (\textsc{LRU}, \textsc{BucketingBypass}, and the CacheFlow-style ones) do not use $\alpha$, and for them $\alpha$ is used only to compute the total cost after the run. This way, all algorithms are evaluated with the same costs, and each of them makes its decisions according to its original rules.

The value of $\alpha$ is the ratio between the cost of an update and the cost of a miss. The CacheFlow paper~\cite{cacheflow} reports roughly 40--50 rule updates per second and a one-way overhead of $100\,\mu$s for a cache miss, which corresponds to a ratio of 200--250. Hence, we use $\alpha=200$ as the default value. To model slower updates, we also evaluate $\alpha \in \{500,800\}$. For \ALG, changing $\alpha$ changes both its saturation threshold and its cost; for the other algorithms, it changes only the cost of a fixed execution.

\subsection{Reproducibility}

A single run of the simulator is determined by a rule file, a~trace specification, an algorithm, a cache capacity, and, where applicable, the value of $\alpha$ or the epoch length. For each run, we store its complete configuration together with all counters. The plots are generated from these stored results, so that a run can be evaluated under a different cost measure without being repeated.

Our base scenario uses the traffic-derived rule tree, the cache ratio of 10\%, and the Zipf exponent of the real trace. Here, the \emph{cache ratio} is the cache capacity divided by the number of rules. In sensitivity experiments, we vary one parameter at a time: the skew, the cache ratio, or the phase length of the synthetic trace. We do not vary the tree depth, as the trees of smaller depth that we could derive contain far fewer rules, and thus the effects of depth and size would be mixed. Synthetic results use five trace seeds (40--44). The seed determines the random permutation of popularities, the generated packets and, for phase-shifted traces, the order of popular subtrees. In synthetic experiments, the randomized algorithm \textsc{BucketingBypass} always uses the seed 42, so the bands in the plots reflect only the randomness of the trace. On the real trace, we average its results over the seeds 40--42.

In each sensitivity run, the first 2.5M packets serve as a~warm-up and are not counted, and we measure the costs on the following 10M packets. (The warm-up length is a~multiple of all tested epoch and phase lengths.) This way, the batch algorithms do not pay for their empty cache in the first epoch, whose share of the total cost would otherwise depend on the length of the trace. In contrast, in the experiments on the full real trace and on synthetic traces of the same length, all algorithms start with an empty cache and we count the whole trace. The plots show the mean over seeds, and the bands show one population standard deviation. For each point of a sensitivity plot, the CacheFlow-style curve uses a~single configuration: among the three variants and the epoch lengths 50K, 100K, 500K, and 2.5M, we choose the one with the smallest mean cost over the five seeds. We use the same configuration for all seeds. As the configuration is chosen on the same traces on which it is evaluated, this choice favors the batch algorithms.

\begin{table*}[t]
  \centering
  \caption{CAIDA NYC packet workload, traffic-derived rule tree, 10\% cache, and $\alpha=200$. Cost uses the formal per-rule accounting.}
  \label{tab:caida_real_results}
  \begin{tabular}{@{}lrrr@{}}
    \toprule
    Algorithm & Hit rate & Rule updates & Cost \\
    \midrule
    \textsc{LRU} (unconstrained) & 0.858 & 167.3M & 33.54G \\
    \textsc{BucketingBypass} & $0.851965\pm0.000001$ & 174.5M & 35.00G \\
    \textsc{DependentSet} (epoch 5M) & 0.811 & 24.7K & 116.5M \\
    \textsc{CoverSet} (epoch 5M) & 0.811 & 24.7K & 116.5M \\
    \textsc{MixedSet} (epoch 5M) & 0.811 & 24.7K & 116.5M \\
    \ALG & 0.708 & 824K & 336.9M \\
    \bottomrule
  \end{tabular}
\end{table*}

The simulator, the scripts for deriving the rule tree and for plotting, the configurations with random seeds, and the aggregated results are available at~\cite{zenodo}. The CAIDA captures are access-restricted, and we cannot redistribute them. Researchers can obtain them freely from CAIDA under its data-use agreement; we provided detailed instructions at~\cite{zenodo} on how to use these captures and our scripts to recreate our experiments. 


\section{Empirical Evaluation}
\label{sec:experimental_evaluation}

In this section, we present the results of our experiments. We are interested in the hit rates that the algorithms achieve on real traffic and in the costs of these hit rates when cache updates are charged. We also examine how the results depend on the parameters of synthetic traces and on the epoch length of the batch algorithms. Recall that the traces contain packet arrivals only, i.e., all requests are positive. We do not compute empirical competitive ratios: computing an optimal offline solution for inputs of this size would require a separate optimization procedure.

The main comparison on the real trace includes \ALG, \textsc{BucketingBypass}, all three CacheFlow-style variants, and \textsc{LRU} as the hit-rate reference. The comparison of the real trace with a synthetic Zipf trace is restricted to \ALG and the three CacheFlow-style variants: its purpose is to check whether a~Zipf model fitted to the whole trace preserves the relation between \ALG and the batch algorithms.

\subsection{Real Backbone Trace}

\cref{tab:caida_real_results} presents the results for the CAIDA trace (589.6M packets) and the cache ratio of 10\%, i.e., the cache holds 754 of the 7,542 rules. For the CacheFlow-style algorithms, we evaluate all three variants with epochs of lengths $\{10\textnormal{K},50\textnormal{K},100\textnormal{K},500\textnormal{K}, 1\textnormal{M},5\textnormal{M},40\textnormal{M},75\textnormal{M}\}$. Each variant attains its lowest cost at 5M packets, and we report all three at this epoch. Their differences are below the precision shown in the table; \textsc{CoverSet} has the lowest exact cost and is therefore used in the summary figure. The epoch is selected on the same trace on which it is evaluated, which is optimistic for the batch algorithms.

The unconstrained \textsc{LRU} achieves the hit rate of 0.858, and \textsc{BucketingBypass} (with the handling of oversized subtrees described in \cref{sec:algorithms_implemented}) achieves 0.852. However, both of them change more than 167M rules, and their costs exceed 33G. Each CacheFlow-style variant loses about 4.1 percentage points of hit rate with respect to \textsc{BucketingBypass}, but its cost is about 300 times smaller. Thus, in our setting, where the cost is dominated by updates, ranking the algorithms by hit rate alone is misleading: the algorithms with the highest hit rates are by far the most expensive ones. The results of \textsc{BucketingBypass} are averaged over the seeds 40--42. Their population standard deviations are $1.25\cdot10^{-6}$ for the hit rate and 0.36M for the cost, and hence the randomness of the algorithm does not affect this conclusion.

\ALG updates 824K rules and its cost is 336.9M. This is 2.9 times the cost of the tuned batch algorithm, but 104 times less than the cost of \textsc{BucketingBypass}. However, the hit rate of \ALG, 0.708, is much lower than the hit rates of both these algorithms. We omit \textsc{Bucketing} from the table, as it has to fetch the five oversized subtrees, which do not fit in the cache. During this run, \ALG performs 546 capacity-triggered phase resets. Each reset evicts a full cache of 754 rules, for 411,684 reset-induced rule evictions in total. Sampling the occupancy every 10K packets gives a mean of 390.3 cached rules, or 51.8\% of the capacity. The frequent full-cache resets together with the moderate mean occupancy help explain its lower hit rate.

All comparisons on the real trace use the cache ratio of 10\%. The dependence on the cache ratio is studied in \cref{fig:sensitivity_summary} on synthetic traffic only, and we do not extrapolate this trend to the real trace.

We note that the ratio of 104 cannot be directly compared with the two-fold reduction reported by Dallot et al.~\cite{dependency-aware-caching}. Their case study uses $\alpha=1$ and stationary synthetic traces of 5,000 requests on a binary tree of 1,023 nodes, while we use $\alpha=200$ and 589.6M real packets. In our cost model, the 174.5M rule changes of \textsc{BucketingBypass} account for 34.9G of its total cost of 35.0G. Hence, when a rule change is much more costly than a miss, their algorithm (specialized to trees) becomes expensive.

\subsection{Real Traffic versus Synthetic Zipf}

Synthetic Zipf traces allow us to control the skew directly. However, a stationary Zipf trace (with independent requests) of the same length and with the same exponent 1.47 behaves quite differently from the real trace, see \cref{tab:real_vs_zipf}. In this table, all CacheFlow-style variants use a fixed epoch of 500K packets. The results on the real trace are deterministic, while the synthetic ones are averaged over the seeds 40--44 (we also give one population standard deviation for the hit rate and the cost). The three real-trace entries only appear identical because of rounding: their hit rates range from 0.829180 to 0.829195, and their costs from 152.398M to 152.439M.

\begin{table*}[t]
  \centering
  \caption{Real CAIDA traffic and stationary Zipf traffic with the same trace length and measured exponent. Cache ratio is 10\% and $\alpha=200$.}
  \label{tab:real_vs_zipf}
  \begin{tabular}{@{}lrrrr@{}}
    \toprule
    & \multicolumn{2}{c}{Real} & \multicolumn{2}{c}{Zipf} \\
    \cmidrule(lr){2-3}\cmidrule(lr){4-5}
    Algorithm & Hit rate & Cost & Hit rate & Cost \\
    \midrule
    \ALG & 0.708 & 336.9M & $0.947\pm0.005$ & $55.0\pm6.2$M \\
    \textsc{DependentSet} (500K) & 0.829 & 152.4M & $0.967\pm0.003$ & $51.2\pm3.6$M \\
    \textsc{CoverSet} (500K) & 0.829 & 152.4M & $0.971\pm0.001$ & $52.2\pm1.5$M \\
    \textsc{MixedSet} (500K) & 0.829 & 152.4M & $0.971\pm0.001$ & $51.5\pm0.9$M \\
    \bottomrule
  \end{tabular}
\end{table*}

For all algorithms, the stationary synthetic trace is easier: all hit rates are at least 0.947, and all costs are between 51M and 55M. In particular, the hit rate of \ALG is 24 percentage points higher than on the real trace. The intervals of one standard deviation around the mean costs of \ALG and the three batch variants overlap, and hence none of these algorithms is a clear winner in terms of cost. The difference between the real and the synthetic results shows that a Zipf distribution fitted globally does not capture the temporal and topological correlations of real traffic.

\subsection{Epoch Sensitivity of Batch Caching}

The performance of the CacheFlow-style algorithms strongly depends on the epoch length. For each epoch length, \cref{tab:epoch_sensitivity} gives the cheapest variant on the CAIDA trace with the cache ratio of 10\%. The trace contains approximately 621K packets per second, which allows us to express epoch lengths also in seconds. With short epochs, the algorithms perform many rule updates. Long epochs reduce the number of updates, but at some point the loss of hit rate becomes too large.

\begin{table}[t]
  \centering
  \caption{Best CacheFlow-style variant at each epoch on the 589.6M-packet CAIDA trace, 10\% cache, $\alpha=200$.}
  \label{tab:epoch_sensitivity}
  \begin{tabular}{@{}lrrrr@{}}
    \toprule
    Epoch & Time & Variant & Hit rate & Cost \\
    \midrule
    100K & 0.16s & Mixed & 0.826 & 481.8M \\
    500K & 0.81s & Cover & 0.829 & 152.4M \\
    1M & 1.61s & Cover & 0.824 & 129.0M \\
    5M & 8.05s & Cover & 0.811 & 116.5M \\
    40M & 64.4s & Dependent & 0.751 & 147.8M \\
    75M & 120.8s & Dependent & 0.698 & 178.3M \\
    \bottomrule
  \end{tabular}
\end{table}

As a function of the epoch length, the cost is U-shaped. The epochs of 10K and 50K packets (omitted from the table, but shown in \cref{fig:empirical_summary}(b)) are worse than the epoch of 100K packets both in hit rate and in cost. Among the epochs in the table, 500K yields the highest hit rate and 5M the lowest cost. For epochs longer than 5M, the savings on updates no longer compensate for the misses caused by a stale cache. The epoch of 75M packets corresponds to approximately two minutes, which is close to the two-minute update schedule reported for CacheFlow~\cite{cacheflow}. In contrast, the cheapest epoch of 5M packets corresponds to an aggressive schedule of 8 seconds. Note that these runs start with an empty cache, and hence a long first epoch additionally incurs the cost of a cold start. Moreover, our simulator does not check whether the rule changes of an 8-second schedule can actually be installed at this rate. Thus, the best configuration of CacheFlow depends on the deployment.

\begin{figure*}[t]
  \centering
  \paperfigure[width=0.7\textwidth,keepaspectratio]{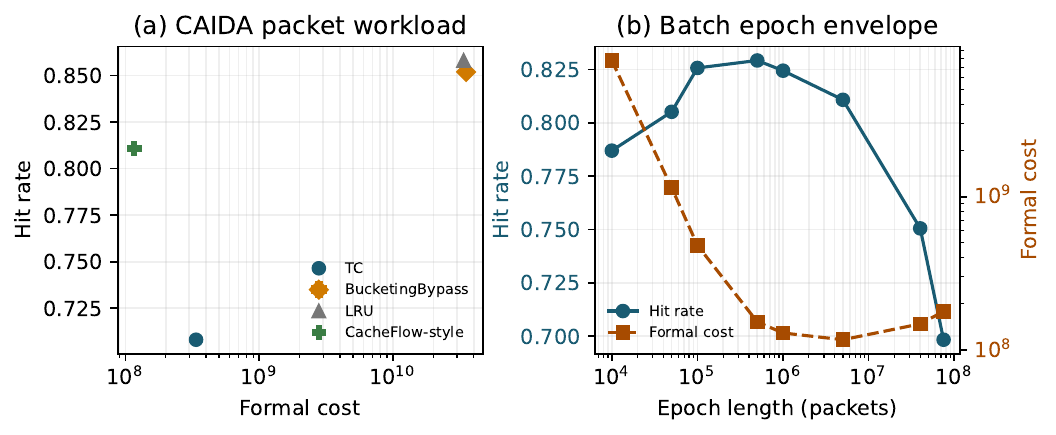}
  \caption{Empirical summary using formal per-rule cost. Panel (a) shows the hit-rate/cost trade-off on the CAIDA packet workload; panel (b) shows the lowest-cost CacheFlow-style variant for each epoch length.}
  \label{fig:empirical_summary}
\end{figure*}

\subsection{Sensitivity Around the Measured Base}

We also ran synthetic experiments around the base scenario derived from the CAIDA trace, i.e., with the traffic-derived rule tree of height~$5$, the cache ratio of 10\%, and the Zipf exponent of approximately 1.47. We vary one parameter at a time: the skew, the cache ratio, or the phase length. \cref{fig:sensitivity_summary} shows both hit rates and costs. At every point, the CacheFlow-style curve shows the configuration with the smallest mean cost, chosen as described in \cref{sec:experimental_implementation}; this is an optimistic estimate for the batch algorithms. The shaded bands show the variability over the five trace seeds.

\begin{figure*}[t]
  \centering
  \paperfigure[width=0.96\textwidth,keepaspectratio]{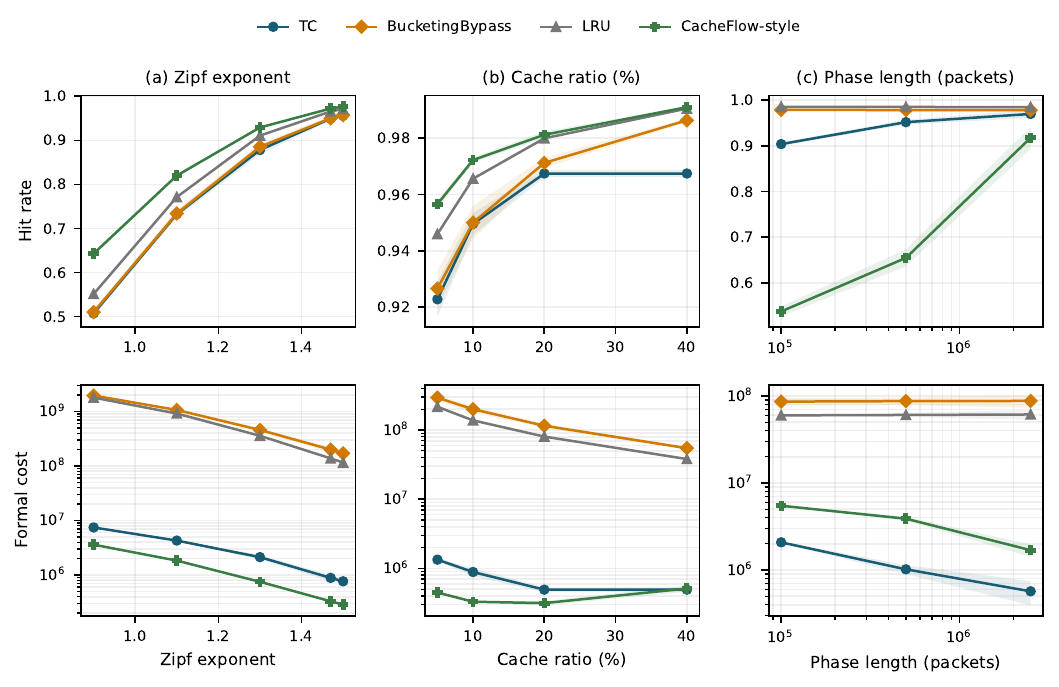}
  \caption{Synthetic sensitivity around the traffic-derived base. The top row shows hit rate and the bottom row formal cost. Phase length is a stress parameter: shorter phases mean faster hot-subtree changes. CacheFlow-style denotes the lowest-mean-cost tested variant and epoch at each point; bands show one population standard deviation over five seeds.}
  \label{fig:sensitivity_summary}
\end{figure*}

The skew substantially affects all results. When the Zipf exponent grows from 0.9 to 1.5, the hit rate of \ALG increases from $0.507\pm0.007$ to $0.956\pm0.004$, and its cost decreases from $7.46\pm0.14$M to $0.77\pm0.06$M. For every skew value, the tuned CacheFlow-style configuration has the smallest mean cost, which decreases from $3.60\pm0.04$M to $0.29\pm0.01$M. Increasing the cache ratio from 5\% to 20\% also lowers the costs of both \ALG and the batch algorithms. At 40\%, the mean cost of \ALG remains $0.49\pm0.04$M, while that of the batch algorithms grows to $0.51\pm0.05$M. As the bands overlap, this small reversal of the means is not significant. The cost of the batch algorithms may grow, as with a larger cache they may change more entries at each reconfiguration. In contrast, \ALG reaches the same mean cost at 20\% and at 40\% once it stops ending its phases because of exceeding the cache capacity.

On phase-shifted traces, the advantage of an online algorithm becomes visible. For phases of 100K, 500K, and 2.5M packets, the costs of \ALG are $2.08\pm0.05$M, $1.02\pm0.12$M, and $0.57\pm0.18$M, respectively, while the corresponding costs of the batch algorithms (in their optimistic configuration) are $5.48\pm0.10$M, $3.89\pm0.20$M, and $1.69\pm0.28$M. Hence, \ALG is cheaper for all tested phase lengths, including the longest one, although both approaches benefit from longer phases. \textsc{BucketingBypass} and \textsc{LRU} keep hit rates above 0.97 on these traces, but as they update the cache upon individual requests, their costs are one to two orders of magnitude larger than that of \ALG. We emphasize that the real trace does not contain such abrupt changes, and the phase-shifted traces serve only to test how well the algorithms adapt.

Finally, we vary $\alpha$ on the real trace. Recall that for \ALG, $\alpha$ affects both the behavior of the algorithm and its cost. For $\alpha=200$, 500, and 800, the hit rate of \ALG is 0.708, 0.698, and 0.693, it updates 824K, 338K, and 213K rules, and its cost is 336.9M, 346.7M, and 351.5M, respectively. That is, a~larger threshold reduces the number of updates at the expense of more misses, and the total cost changes only slightly. The overall picture remains the same.


\section{Conclusions}\label{sec:conclusion}

This paper defines a novel variant of online paging which finds applications in the context of IP routing networks where forwarding rules can be cached. We present a deterministic online algorithm that achieves a provably competitive trade-off between the benefit of caching and update costs.

Our experiments show that hit rate and update cost have to be considered together. On the CAIDA trace, the algorithms that update the cache upon individual misses achieve the highest hit rates, but their cost is over two orders of magnitude larger than that of the batch algorithms with a tuned epoch. The hit rate of \ALG is lower, and its cost is 2.9 times that of the tuned batch algorithms. On stationary synthetic traffic, the one-standard-deviation intervals around the costs of \ALG and of the batch algorithms with a fixed epoch overlap, while the batch algorithms with a tuned epoch are cheaper for all tested skews. On traces with abrupt changes of the popular subtree, \ALG is the cheapest tested algorithm for all phase lengths. These conclusions concern positive requests only; evaluating negative requests requires traces that also contain updates of routing rules.

We believe that our work opens interesting directions for future research. The optimality of the derived result is open, and we conjecture that the true competitive ratio does not depend on the tree height. In particular, primal-dual approaches that were successfully applied for other caching problems~\cite{young-paging-greedy-dual,generalized-caching-optimal,generalized-caching-bansal} may turn out to be useful also for the considered variant. Empirically, the next steps are an update-aware workload and a scale-matched rule tree derived from an independent routing-table snapshot.


\section*{Acknowledgments}

The authors would like to thank Fred Baker from Cisco, Moti Medina from the Bar-Ilan University and Paweł Gawrychowski from University of Wrocław for useful inputs.


\bibliography{references}

\begin{thebibliography}{10}

\bibitem{paging-optimal-easy}
Dimitris Achlioptas, Marek Chrobak, and John Noga.
\newblock Competitive analysis of randomized paging algorithms.
\newblock {\em Theoretical Computer Science}, 234(1--2):203--218, 2000.

\bibitem{generalized-caching-optimal}
Anna Adamaszek, Artur Czumaj, Matthias Englert, and Harald R{\"{a}}cke.
\newblock An \emph{O}(log \emph{k})-competitive algorithm for generalized
  caching.
\newblock In {\em Proc. 23rd ACM-SIAM Symp. on Discrete Algorithms (SODA)},
  pages 1681--1689, 2012.

\bibitem{generalized-caching-bansal}
Nikhil Bansal, Niv Buchbinder, and Joseph Naor.
\newblock Randomized competitive algorithms for generalized caching.
\newblock {\em SIAM Journal on Computing}, 41(2):391--414, 2012.

\bibitem{multi-level-paging}
Nikhil Bansal, Joseph~(Seffi) Naor, and Ohad Talmon.
\newblock Efficient online weighted multi-level paging.
\newblock In {\em Proc. 33rd ACM Symp. on Parallelism in Algorithms and
  Architectures (SPAA)}, pages 94--104. {ACM}, 2021.

\bibitem{tree-caching-spaa}
Marcin Bienkowski, Jan Marcinkowski, Maciej Pacut, Stefan Schmid, and
  Aleksandra Spyra.
\newblock Online tree caching.
\newblock In {\em Proc. 29th ACM Symp. on Parallelism in Algorithms and
  Architectures (SPAA)}, pages 329--338. {ACM}, 2017.

\bibitem{zenodo}
Marcin Bienkowski, Jan Marcinkowski, Maciej Pacut, Stefan Schmid, Aleksandra
  Spyra, and Bolong Tang.
\newblock Online tree caching: evaluation artifact.
\newblock \url{https://doi.org/10.5281/zenodo.23142566}, October 2026.

\bibitem{fib-icdcs}
Marcin Bienkowski, Nadi Sarrar, Stefan Schmid, and Steve Uhlig.
\newblock Competitive {FIB} aggregation without update churn.
\newblock In {\em Proc. 34th IEEE Int. Conf. on Distributed Computing Systems
  (ICDCS)}, pages 607--616, 2014.

\bibitem{fib-ton}
Marcin Bienkowski, Nadi Sarrar, Stefan Schmid, and Steve Uhlig.
\newblock Online aggregation of the forwarding information base: Accounting for
  locality and churn.
\newblock {\em IEEE/ACM Transactions on Networking}, 26(1):591--604, 2018.

\bibitem{fib-sirocco}
Marcin Bienkowski and Stefan Schmid.
\newblock Competitive {FIB} aggregation for independent prefixes: Online ski
  rental on the trie.
\newblock In {\em Proc. 20th Int. Colloq. on Structural Information and
  Communication Complexity (SIROCCO)}, pages 92--103, 2013.

\bibitem{restricted-caching}
Mark Brehob, Richard~J. Enbody, Eric Torng, and Stephen Wagner.
\newblock On-line restricted caching.
\newblock {\em Journal of Scheduling}, 6(2):149--166, 2003.

\bibitem{matroid-caching}
Niv Buchbinder, Shahar Chen, and Joseph Naor.
\newblock Competitive algorithms for restricted caching and matroid caching.
\newblock In {\em Proc. 22th European Symp. on Algorithms (ESA)}, pages
  209--221, 2014.

\bibitem{caida-nyc-trace}
{CAIDA}.
\newblock The {CAIDA UCSD} anonymized internet traces -- {January} 17, 2019.
\newblock \url{https://www.caida.org/catalog/datasets/passive_dataset/}, 2019.
\newblock Equinix NYC monitor, direction A.

\bibitem{double-coverage}
Marek Chrobak, Howard~J. Karloff, Thomas~H. Payne, and Sundar Vishwanathan.
\newblock New results on server problems.
\newblock {\em SIAM Journal on Discrete Mathematics}, 4(2):172--181, 1991.

\bibitem{steve-myth}
Luca Cittadini, Wolfgang Muhlbauer, Steve Uhlig, Randy Bushy, Pierre Francois,
  and Olaf Maennel.
\newblock Evolution of internet address space deaggregation: myths and reality.
\newblock {\em IEEE J.Sel. A. Commun.}, 28(8):1238--1249, 2010.

\bibitem{block-aware-caching}
Christian Coester, Roie Levin, Joseph~(Seffi) Naor, and Ohad Talmon.
\newblock Competitive algorithms for block-aware caching.
\newblock In {\em Proc. 34th ACM Symp. on Parallelism in Algorithms and
  Architectures (SPAA)}, pages 161--172. {ACM}, 2022.

\bibitem{dependency-aware-caching}
Julien Dallot, Amirmehdi~Jafari Fesharaki, Maciej Pacut, and Stefan Schmid.
\newblock Dependency-aware online caching.
\newblock In {\em Proc. IEEE Int. Conf. on Computer Communications (INFOCOM)},
  pages 871--880, 2024.

\bibitem{non-linear-paging}
Ilan Doron-Arad and Joseph~(Seffi) Naor.
\newblock {Non-Linear Paging}.
\newblock In {\em Proc. 51st Int. Colloq. on Automata, Languages and
  Programming (ICALP)}, volume 297, pages 57:1--57:19, 2024.

\bibitem{ortc}
Richard~P. Draves, Christopher King, Srinivasan Venkatachary, and Brian~D.
  Zill.
\newblock Constructing optimal {IP} routing tables.
\newblock In {\em Proc. 18th IEEE Int. Conf. on Computer Communications
  (INFOCOM)}, pages 88--97, 1999.

\bibitem{caching-rejection-penalties}
Leah Epstein, Csan{\'{a}}d Imreh, Asaf Levin, and Judit Nagy{-}Gy{\"{o}}rgy.
\newblock Online file caching with rejection penalties.
\newblock {\em Algorithmica}, 71(2):279--306, 2015.

\bibitem{paging-mark}
Amos Fiat, Richard~M. Karp, Michael Luby, Lyle~A. McGeoch, Daniel~D. Sleator,
  and Neal~E. Young.
\newblock Competitive paging algorithms.
\newblock {\em Journal of Algorithms}, 12(4):685--699, 1991.

\bibitem{fib-updates}
Pierre Fran{\c{c}}ois, Clarence Filsfils, John Evans, and Olivier Bonaventure.
\newblock Achieving sub-second {IGP} convergence in large {IP} networks.
\newblock {\em ACM SIGCOMM Computer Communication Review}, 35(3):35--44, 2005.

\bibitem{tcam-expensive-updates}
Danny~Yuxing Huang, Ken Yocum, and Alex~C. Snoeren.
\newblock High-fidelity switch models for software-defined network emulation.
\newblock In {\em Proc. 2nd{ACM SIGCOMM} Workshop on Hot Topics in Software
  Defined Networking (HotSDN)}, pages 43--48, 2013.

\bibitem{bgp-routeviews}
Geoff Huston.
\newblock {BGP} statistics from route-views data.
\newblock \url{https://bgp.potaroo.net/as6447/}, 2026.
\newblock Accessed: 2026-08-24.

\bibitem{paging-irani}
Sandy Irani.
\newblock Page replacement with multi-size pages and applications to web
  caching.
\newblock {\em Algorithmica}, 33(3):384--409, 2002.

\bibitem{mms}
Elliott Karpilovsky, Matthew Caesar, Jennifer Rexford, Aman Shaikh, and
  Jacobus~E. van~der Merwe.
\newblock Practical network-wide compression of {IP} routing tables.
\newblock {\em {IEEE} Transactions on Network and Service Management},
  9(4):446--458, 2012.

\bibitem{cacheflow}
Naga Katta, Omid Alipourfard, Jennifer Rexford, and David Walker.
\newblock Cacheflow: Dependency-aware rule-caching for software-defined
  networks.
\newblock In {\em Proc. ACM Symposium on SDN Research (SOSR)}, 2016.

\bibitem{route-caching-flat}
Changhoon Kim, Matthew Caesar, Alexandre Gerber, and Jennifer Rexford.
\newblock Revisiting route caching: The world should be flat.
\newblock In {\em Proc. 10th Int. Conf. on Passive and Active Network
  Measurement (PAM)}, pages 3--12, 2009.

\bibitem{prefix-caching}
Huan Liu.
\newblock Routing prefix caching in network processor design.
\newblock In {\em Proc. 10th Int. Conf. on Computer Communications and Networks
  (ICCCN)}, pages 18--23, 2001.

\bibitem{fib-caching-non-overlapping}
Yaoqing Liu, Vince Lehman, and Lan Wang.
\newblock Efficient {FIB} caching using minimal non-overlapping prefixes.
\newblock {\em Computer Networks}, 83:85--99, 2015.

\bibitem{fib-compression-fifa}
Yaoqing Liu, Beichuan Zhang, and Lan Wang.
\newblock {FIFA:} fast incremental {FIB} aggregation.
\newblock In {\em Proc. 32nd IEEE Int. Conf. on Computer Communications
  (INFOCOM)}, pages 1213--1221, 2013.

\bibitem{fib-compression-globecom10}
Yaoqing Liu, Xin Zhao, Kyuhan Nam, Lan Wang, and Beichuan Zhang.
\newblock Incremental forwarding table aggregation.
\newblock In {\em Proc. Global Communications Conference (GLOBECOM)}, pages
  1--6, 2010.

\bibitem{fib-compression-infocom13}
Layong Luo, Gaogang Xie, Kav{\'{e}} Salamatian, Steve Uhlig, Laurent Mathy, and
  Yingke Xie.
\newblock A trie merging approach with incremental updates for virtual routers.
\newblock In {\em Proc. 32nd IEEE Int. Conf. on Computer Communications
  (INFOCOM)}, pages 1222--1230, 2013.

\bibitem{paging-optimal-difficult}
Lyle~A. McGeoch and Daniel~D. Sleator.
\newblock A strongly competitive randomized paging algorithm.
\newblock {\em Algorithmica}, 6(6):816--825, 1991.

\bibitem{companion-caching}
Manor Mendel and Steven~S. Seiden.
\newblock Online companion caching.
\newblock {\em Theoretical Computer Science}, 324(2--3):183--200, 2004.

\bibitem{fib-sigcomm}
G{\'{a}}bor R{\'{e}}tv{\'{a}}ri, J{\'{a}}nos Tapolcai, Attila Kor{\"{o}}si,
  Andr{\'{a}}s Majd{\'{a}}n, and Zal{\'{a}}n Heszberger.
\newblock Compressing {IP} forwarding tables: towards entropy bounds and
  beyond.
\newblock In {\em Proc. ACM SIGCOMM Conference}, pages 111--122, 2013.

\bibitem{fibium-zipf}
Nadi Sarrar, Steve Uhlig, Anja Feldmann, Rob Sherwood, and Xin Huang.
\newblock Leveraging {Zipf}'s law for traffic offloading.
\newblock {\em ACM SIGCOMM Computer Communication Review}, 42(1):16--22, 2012.

\bibitem{competitive-analysis}
Daniel~D. Sleator and Robert~E. Tarjan.
\newblock Amortized efficiency of list update and paging rules.
\newblock {\em Communications of the ACM}, 28(2):202--208, 1985.

\bibitem{tcam-expensive}
Ed~Spitznagel, David~E. Taylor, and Jonathan~S. Turner.
\newblock Packet classification using extended tcams.
\newblock In {\em Proc. 11th IEEE Int. Conf. on Network Protocols (ICNP)},
  pages 120--131, 2003.

\bibitem{fib-compression-two-dimensional}
Subhash Suri, Tuomas Sandholm, and Priyank~Ramesh Warkhede.
\newblock Compressing two-dimensional routing tables.
\newblock {\em Algorithmica}, 35(4):287--300, 2003.

\bibitem{fib-compression-smalta}
Zartash~Afzal Uzmi, Markus~E. Nebel, Ahsan Tariq, Sana Jawad, Ruichuan Chen,
  Aman Shaikh, Jia Wang, and Paul Francis.
\newblock {SMALTA:} practical and near-optimal {FIB} aggregation.
\newblock In {\em Proc. 7th Int. Conf. on Emerging Networking Experiments and
  Technologies (CoNEXT)}, 2011.

\bibitem{young-paging-greedy-dual}
Neal~E. Young.
\newblock The k-server dual and loose competitiveness for paging.
\newblock {\em Algorithmica}, 11(6):525--541, 1994.

\bibitem{young-paging-landlord}
Neal~E. Young.
\newblock On-line file caching.
\newblock {\em Algorithmica}, 33(3):371--383, 2002.

\bibitem{fib-compression-infocom10}
Xin Zhao, Yaoqing Liu, Lan Wang, and Beichuan Zhang.
\newblock On the aggregatability of router forwarding tables.
\newblock In {\em Proc. 29th IEEE Int. Conf. on Computer Communications
  (INFOCOM)}, pages 848--856, 2010.

\end{thebibliography}

\appendices


\section{Proof of \cref*{lem:no_over-requested_changesets}}
\label{sec:proof_of_lemma_1}

Before proving \cref{lem:no_over-requested_changesets}, we present the following technical claim.

\begin{claim}
\label{cla:inductive_properties}
For any phase $P$, the following invariants hold for any time $t > \beP$:
\begin{enumerate}
\item $\cnt_{t-1}(X) < |X| \cdot \alpha$ for a valid changeset $X$ for $C_t$,\hspace{-1em}
\label[invariant]{item:prop1}
\item $\cnt_t(X) \leq |X| \cdot \alpha$ for a valid changeset $X$ for $C_t$,
\label[invariant]{item:prop2}
\item any valid changeset $X$ for $C_t$ with $\cnt_t(X) = |X| \cdot \alpha$ contains the node requested at round $t$.
\label[invariant]{item:propmid}
\end{enumerate}
\end{claim}

\begin{IEEEproof}
First observe that \cref{item:prop1} (for time $t$) along with the fact that round $t$ contains only one request immediately implies that $\cnt_t(X) \leq \cnt_{t-1}(X) + 1 \leq (|X| \cdot \alpha - 1) + 1 = |X| \cdot \alpha$, i.e., \cref{item:prop2} for time~$t$. Furthermore, the equality may hold only for changesets containing the node requested at round $t$, which implies \cref{item:propmid} for time $t$.

It remains to show that \cref{item:prop1} holds for any step $t > \beP$. It is trivially true for $t = \beP+1$ as $\cnt_{t-1}(X) = 0$ then. Let $t+1$ be the earliest time in phase~$P$ for which \cref{item:prop1} does not hold; we will then show a~contradiction with the definition of \ALG or a~contradiction with other invariants at time $t$. That is, we assume that there exists a positive changeset $X$ for $C_{t+1}$ such that $\cnt_t(X) \geq |X| \cdot \alpha$ (the proof for a negative changeset is analogous). Note that \ALG must have performed an action (fetch or eviction) at time~$t$ as otherwise $X$ would be also a changeset for $C_t = C_{t+1}$ with $\cnt_t(X) \geq |X| \cdot \alpha$, which means that some changeset should have been applied by \ALG at time $t$. We consider two cases.

If \ALG fetches a positive changeset $Y$ at time $t$, $C_{t+1} = C_t \uplus Y$ and $\cnt_t(Y) = |Y|\cdot\alpha$. Then, $Y \uplus X$ is a changeset for $C_t$, and $\cnt_t(Y \uplus X) \geq |Y \uplus X| \cdot \alpha$. This contradicts the maximality property of set~$Y$ chosen at time~$t$ by~\ALG.

If \ALG evicts a negative changeset $Y$ at time $t$, $C_{t+1} = C_t \setminus Y$. \cref{item:prop2} and the definition of \ALG imply $\cnt_t(Y) = |Y| \cdot \alpha$, and thus, by \cref{item:propmid}, $Y$ contains the node requested at round $t$. As the counters of $Y$ are reset at time~$t$, $\cnt_t(X \setminus Y) = \cnt_t(X) \geq |X| \cdot \alpha \geq |X \setminus Y| \cdot \alpha$. In particular, $X \setminus Y \neq \emptyset$, as otherwise $\cnt_t(X) = 0$. By \cref{item:prop2}, $\cnt_t(X \setminus Y) \leq |X \setminus Y| \cdot \alpha$, and hence $\cnt_t(X \setminus Y) = |X \setminus Y| \cdot \alpha$. This contradicts \cref{item:propmid} as $X \setminus Y$ cannot contain the node requested at round $t$ (because $Y$ contains this node).
\end{IEEEproof}

\begin{IEEEproof}[Proof of \cref*{lem:no_over-requested_changesets}]
The inequality $\cnt_t(X) \leq |X| \cdot \alpha$ is equivalent to \cref{item:prop2} of \cref{cla:inductive_properties}. Assume now that $X$ is applied at time $t$. By the definition of \ALG, $\cnt_t(X) \geq |X| \cdot \alpha$, and thus $\cnt_t(X) = |X| \cdot \alpha$, i.e., \cref{lemit:2} follows. Then, \cref{item:propmid} of \cref{cla:inductive_properties} implies \cref{lemit:1}. Finally, \cref{item:prop1} of \cref{cla:inductive_properties} for time $t+1$ is equivalent to \cref{lemit:3}.

To show \cref{lemit:4}, observe that the changeset $X$ applied at time $t$ cannot be a disjoint union of two (or more) valid changesets $X_1$ and $X_2$. By \cref{lemit:2}, $|X| \cdot \alpha = \cnt_t(X) = \cnt_t(X_1) + \cnt_t(X_2)$. If $\cnt_t(X_1) < |X_1| \cdot \alpha$ or $\cnt_t(X_2) < |X_2| \cdot \alpha$, then $\cnt_t(X_1) + \cnt_t(X_2) < (|X_1| + |X_2|) \cdot \alpha = |X| \cdot \alpha$, a contradiction. Therefore, $\cnt_t(X_1) = |X_1| \cdot \alpha$ and $\cnt_t(X_2) = |X_2| \cdot \alpha$. But then \cref{item:propmid} of \cref{cla:inductive_properties} would imply that both $X_1$ and $X_2$ contain the node requested at round~$t$, which is a~contradiction as they are disjoint.

Therefore, if $X$ is a positive changeset applied at $t$, then $X$ is a~single tree cap of a tree from subforest $C_{t+1}$, and likewise if $X$ is negative, then $X$ is a~single tree cap of a~tree from subforest $C_t$.
\end{IEEEproof}


\section{Minimizing Forwarding Tables Using Tree Caching}
\label{sec:bisimulation}

In this section, we present a formal argument showing why we can use any $q$-competitive online algorithm $A_T$ for the tree caching problem to obtain a $2 q$-competitive online algorithm~$A$ that minimizes forwarding tables.

Namely, we take any input $I$ for the latter problem and create, in online fashion, an input $I_T$ for the tree caching problem in a way described in \cref{sec:motivation}. For any solution for~$I_T$, we may replay its actions (fetches and evictions) on~$I$ and vice versa. However, there is one place where these solutions may have different costs. Recall that an update of a~rule stored at node~$v$ in~$I$ is mapped to a \emph{chunk} of $\alpha$ negative requests to~$v$ in~$I_T$. It is then possible that an algorithm for $I_T$ modifies the cache \emph{during} a chunk. An~algorithm that never performs such an action is called \emph{canonical}.

To alleviate this issue, we first note that any algorithm $B$ for $I_T$ can be transformed into a canonical solution $B'$ by postponing all cache modifications that occur during some chunk to the time right after it. Such a transformation may increase the cost of a~solution on a~chunk by at most $\alpha$ and such an increase occurs only when $B$ modifies a cache within this chunk. Hence, the additional cost of transformation can be mapped to the already existing cost of $B$, and thus the cost of $B'$ is larger than that of $B$ by at most a factor of $2$.

Furthermore, note that there is a natural cost-preserving bijection between solutions to~$I$ and canonical solutions to~$I_T$ (both solutions perform the same cache modifications). Hence, the algorithm $A$ for $I$ runs $A_T$ on~$I_T$, transforms it in an~online manner into the canonical solution $A'_T(I_T)$, and replays its cache modifications on $I$. Then, $A(I) = A'_T(I_T) \leq 2 \cdot A_T(I_T) \leq 2 q \cdot \OPT(I_T) \leq 2 q \cdot \OPT(I)$.

The second inequality follows immediately by the $q$-com\-pe\-ti\-ti\-ve\-ness of $A_T$. The third inequality follows by replaying cache modifications as well, but this time we take solution~$\OPT(I)$ and replay its actions on $I_T$, creating a canonical (not necessarily optimal) solution of the same cost.


\section{Lower Bound on the Competitive Ratio}
\label{sec:lower-bound-on-the-problem}

\begin{theorem}
For any $\alpha \geq 1$, the competitive ratio of any deterministic online algorithm for the online tree caching problem is at least $\Omega(\kALG/(\kALG-\kOPT+1))$.
\end{theorem}

\begin{IEEEproof}
We will assume that in the tree caching problem, evictions are free (this changes the cost by at most a factor of two). We consider a tree whose leaves correspond to the set of all pages in the paging problem. The rest of the tree will be irrelevant.

For any input sequence $I$ for the paging problem, we may create a sequence $I_\textnormal{T}$ for tree caching, where a request to a page is replaced by $\alpha$ requests to the corresponding leaf. Now, we claim that any solution $A$ for~$I$ of cost $c$ can be transformed, in an online manner, into a~solution $A_\textnormal{T}$ for $I_\textnormal{T}$ of cost $\Theta(\alpha \cdot c)$ and vice versa.

If upon a request $r$, an algorithm $A$ fetches $r$ to the cache and evicts some pages, then $A_\textnormal{T}$ bypasses $\alpha$ corresponding requests to leaf $r$, fetches $r$ afterwards and evicts the corresponding leaves, paying $O(\alpha)$ times the cost of $A$. By doing it iteratively, $A_\textnormal{T}$ ensures that its cache is equivalent to that of~$A$. In particular, a request free for $A$ is also free for $A_\textnormal{T}$.

Now take any algorithm $A_\textnormal{T}$ for $I_\textnormal{T}$. It can be transformed to the algorithm $A_\textnormal{T}'$ that (i)~keeps only leaves of the tree in the cache and (ii) performs actions only at times that are multiples of $\alpha$ (losing at most a constant factor in comparison to $A_\textnormal{T}$). Then, fix any chunk of $\alpha$ requests to some leaf $r'$ immediately followed by some fetches and evictions of leaves performed by $A_\textnormal{T}'$. Upon seeing the corresponding request~$r'$ in $I$, the algorithm $A$ performs fetches and evictions on the corresponding pages. If $A_\textnormal{T}'$ bypasses the chunk without fetching $r'$, then $A$ fetches $r'$ (evicting an arbitrary page $p$ if its cache is full), serves the request, and restores its previous cache by evicting $r'$ and fetching $p$ back. This costs $O(1)$, whereas $A_\textnormal{T}'$ pays $\alpha$ for the chunk. In effect, the cost of $A$ is $O(1/\alpha)$ times the cost of $A_\textnormal{T}$.

The bidirectional reduction described above preserves competitive ratios up to a constant factor. Hence, applying the adversarial strategy for the paging problem that enforces the competitive ratio $R = \kALG/(\kALG-\kOPT+1)$ \cite{competitive-analysis} immediately implies the lower bound of $\Omega(R)$ on the competitive ratio for the tree caching problem.
\end{IEEEproof}

\end{document}